\documentclass[11pt]{article}
\usepackage{amsfonts}
\usepackage{amsmath,amsthm}

\usepackage{mathrsfs}
\usepackage{cite}
\usepackage{color}
\usepackage[colorlinks,linkcolor=blue,citecolor=blue]{hyperref}
\usepackage{graphicx,subfigure}
\usepackage[numbers,sort&compress]{natbib}
\usepackage[top=1 in,bottom=1 in,left=1.0 in,right=1.0 in]{geometry}

\newcommand{\bblu}{\begin{color}{blue}}
\newcommand{\bred}{\begin{color}{red}}
\newcommand{\ecl}{\end{color}}
\numberwithin{equation}{section}

\newtheorem{proposition}{Proposition}[section]
\newtheorem{lemma}{Lemma}[section]
\newtheorem{theorem}{Theorem}[section]

\def\t#1{\widetilde{#1}}

\def\b#1{\overline{#1}}

\begin{document}

\title{Algebro-geometric integration to the discrete Chen-Lee-Liu system}
\author{Xiaoxue Xu$^{1}$, Decong Yi$^{1}$, Xing Li$^{2}$, Da-jun Zhang$^{3,4}$\footnote{
Corresponding author. E-mail:  djzhang@staff.shu.edu.cn}
\\
{\small $^{1}$School of Mathematics and Statistics, Zhengzhou University, Zhengzhou, Henan 450001, P.R. China}\\
{\small $^2$School of Mathematics and Statistics, Jiangsu Normal University, Xuzhou 221116, China}\\
{\small $^3$Department  of Mathematics, Shanghai University, Shanghai 200444, China}\\
{\small $^{4}$Newtouch Center for Mathematics of Shanghai University,  Shanghai 200444, China}
}

\date{\today}

\maketitle

\begin{abstract}
Algebro-geometric solutions for the discrete Chen-Lee-Liu (CLL) system are derived in this paper.
We construct a nonlinear integrable symplectic map which is used to define discrete phase flows.
Compatibility of the maps with different parameters
gives rise to the discrete CLL system  whose solutions (discrete potentials)
can be formulated through the  discrete phase flows.
Baker-Akhiezer functions are introduced and their asymptotic behaviors are analyzed.
Consequently, we are able to reconstruct the discrete potentials in terms of the Riemann theta functions.
These results can be extended to 3-dimensional case and
algebro-geometric solutions of the discrete modified Kadomtsev-Petviashvili equation are obtained.
Some solutions of genus one case are illustrated.
\\[2mm]
\\[1mm]
\end{abstract}

%\tableofcontents
%
%\newpage

\section{Introduction }\label{sec-1}

It is well-known that many integrable equations admit quasi-periodic solutions,
also known as finite-gap, finite-genus, finite-band or algebro-geometric solutions,
which can be derived by using ``finite-gap integration'' technique developed in 1970s
by Novikov, Matveev and their collaborators
\cite{Dubrovin-1975,DN-1975,IM-1975a,IM-1975b,Krichever-1977, Novikov-1974},
(also see \cite{BBEIM-1994} and \cite{Matveev} and the references therein).
An alternative way to get these solutions is based on the idea of nonlinearizing Lax pairs \cite{Cao1990,Cao-G-1990},
leading to variable-separated finite-dimensional Hamiltonian systems,
which allows to derive finite-gap solutions either numerically \cite{RCW-1995}
or exactly \cite{Cao-WG-1999,Cao-WG-1999b,Zhou-1997}.
This approach does not rely on the analysis of the spectrum associated with periodic potentials
and has been applied to a lot of continuous as well as differential-difference integrable systems,
e.g. \cite{GDZ-2007,GWH-2013,GZD-2014}.
In the differential-difference case, the evolution with respect to the discrete independent variable
is formulated by the iteration of a symplectic map \cite{Cao-WG-1999b}.
Such an idea has been extended to fully discrete integrable systems
by Cao, Xu and their collaborators
\cite{CaoX-JPA-2012,CXZ-2020,CZ2012a,CZ2012b,XCN-2021,XCZ-2022,XCZ-JNMP-2020,XJN-2021}.
In practice, however, the fully discrete case is much more complicated.
For some fully discrete soliton equations, a continuous dummy independent  variable is still needed
to recover the discrete potentials, e.g.\cite{CaoX-JPA-2012,CXZ-2020,XJN-2021},
while for some other equations auxiliary independent variables are not necessary,
e.g.\cite{CZ2012a,XCZ-JNMP-2020,XCN-2021,CZ2012b,XJN-2021,XCZ-2022}.
In both cases, the key step (and also the most important step) is to construct
symplectic maps which are associated with the Lax pairs
of the investigated discrete integrable systems.

The purpose of this paper is to derive quasi-periodic solutions for
the discrete Chen-Lee-Liu (CLL) system \cite{DJM-1983}
\begin{subequations}\label{eq:1.1}
\begin{align}
&\Xi_1 \doteq \beta _{2}^{2}\b{a}-\beta _{1}^{2}\t{a}
+a(\beta _{1}^{2}-\beta _{2}^{2})+a\b{\t{b}}(\overline{a}-\t{a})=0,\label{eq:1.1a} \\
&\Xi_2 \doteq \beta _{1}^{2}\b{b}-\beta _{2}^{2}\t{b}
+\overline{\t{b}}(\beta _{2}^{2}-\beta _{1}^{2})+a\b{\widetilde{b}}(\b{b}-\t{b})=0.\label{eq:1.1b}
\end{align}
\end{subequations}
Here $a=a(m_1,m_2)$ and $b=b(m_1,m_2)$ are functions of $(m_1,m_2)$ defined on $\mathbb{Z}^2$;
$\beta _{1}$ and $\beta_2$  are lattice parameters of $m_1$- and $m_2$-direction, respectively;
tilde and hat serve as shift notations, e.g.
\[\t{a}=a(m_{1}+1,m_{2}),~~ \b{a}=a(m_{1},m_{2}+1),
~~ \t{\b{a}}=a(m_{1}+1,m_{2}+1).\]
Integrability of this system has been described in \cite{DJM-1983} with a Lax pair
\begin{subequations}\label{Lax}
\begin{align}\label{Lax-1}
& \t\Phi=
\left(\begin{array}{cc}
-\eta^2+\beta_1^2+ a \t b &      a\\
  \eta^2 \t b &  \beta_1^2
\end{array}\right)
\Phi,\\
& \b\Phi=
\left(\begin{array}{cc}
-\eta^2+\beta_2^2 + a \b b&      a\\
  \eta^2 \b b &  \beta_2^2
\end{array}\right)
\Phi,
\end{align}
\end{subequations}
where $\Phi=(\phi_1, \phi_2)^T$ and $\eta$ serves as the spectral parameter.
In the continuum limits, the system \eqref{eq:1.1} gives rise to the unreduced continuous
derivative nonlinear Schr\"odinger (NLS) system of the CLL type:
\begin{subequations}\label{eq:0.000}
\begin{align}
&a_{y}+a_{xx}+2aba_{x}=0,\label{eq:0.000a} \\
&b_{y}-b_{xx}+2abb_{x}=0,\label{eq:0.000b}
\end{align}
\end{subequations}
where $x$ and $y$ are Miwa variables defined as
\begin{equation}\label{Miwa}
x=\frac{m_1}{\beta_1^2}+\frac{m_2}{\beta_2^2},~~
y=\frac{m_1}{2\beta_1^4}+\frac{m_2}{2\beta_2^4},
\end{equation}
and the related spectral problem is \cite{DJM-1983}
\begin{equation}\label{c-CLL-sp-app}
\Phi_x=
\left(\begin{array}{cc} -\eta^2+a b  & \eta a \\
 \eta  b & 0 \end{array}\right)\Phi.
\end{equation}
By reduction $b=a^*$, the system \eqref{eq:0.000} yields the CLL derivative NLS equation  \cite{CLL}:
\begin{equation}\label{CLL}
\mathrm{i} a_{y}+a_{xx}+2\mathrm{i} |a|^2a_{x}=0,
\end{equation}
where $|a|^2=a a^*$, $*$ stands for complex conjugate, $\mathrm{i}^2=-1$,
and we have replace $y\to- \mathrm{i} y, x\to- \mathrm{i} x$.
Note that the unreduced CLL system \eqref{eq:0.000} and the spectral problem \eqref{c-CLL-sp-app}
(up to some gauge transformation) are the consequence of the constraint of the squared eigenfunction symmetry
of the modified Kadomtsev-Petviashvili (mKP) system \cite{Chen-2002},
while the discrete CLL system \eqref{eq:1.1} and the spectral problem \eqref{Lax-1} (up to some gauge transformation)
are the discrete analogues that are related to the differential-difference mKP system, see \cite{LZZ}.
In addition, \eqref{Lax-1} acts as a Darboux transformation of \eqref{c-CLL-sp-app}
to transform $(\Phi, a,b)$ to $(\t \Phi, a, \t b)$ \cite{LZZ}.

In a recent paper \cite{XCZ-2022} we have derived quasi-periodic solutions for the discrete derivative NLS system
of the Kaup-Newell type:
\begin{subequations}\label{KN}
\begin{align}
&\frac{1}{2}(1+\sqrt{1-4u\t{v}}\,)(\t{u}+\beta_1^2u)
-\frac{1}{2}(1+\sqrt{1-4u\b{v}}\,) (\b{u}+\beta_2^2u)-(\t{u}\t{v}-\b{u}\,\b{v})u=0, \\
&\frac{1}{2}(1+\sqrt{1-4\t{u}\b{\t{v}}}\,)(\t{v}+\beta_2^2\b{\t{v}})
-\frac{1}{2}(1+\sqrt{1-4\b{u}\b{\t{v}}}\,)(\b{v}+\beta_1^2\b{\t{v}})-(\t{u}\t{v}-\b{u}\,\b{v})\b{\t{v}}=0,
\end{align}
\end{subequations}
whose continuous counterpart in terms of the Miwa variable \eqref{Miwa} reads (cf.\eqref{eq:0.000})
\begin{equation}
  u_y+u_{xx}+2(u^2v)_x=0,~\quad   v_y-v_{xx}+2(uv^2)_x=0.
\end{equation}
In addition to \eqref{KN}, the discrete mKP equation \cite{NCW-LNP-1985,Nijhoff},
\begin{equation}\label{eq:A-1-1}
\beta_1^2(e^{-\overline{\widetilde{W}}+\overline{W}}-e^{-\widehat{\widetilde{W}}+\widehat{W}})
+\beta_2^2(e^{-\widehat{\overline{W}}+\widehat{W}}-e^{-\widetilde{\overline{W}}+\widetilde{W}})
+\beta_3^2(e^{-\widetilde{\widehat{W}}+\widetilde{W}}-e^{-\overline{\widehat{W}}+\overline{W}})=0,
\end{equation}
was solved in \cite{XCZ-2022} by using the connections between \eqref{KN} and \eqref{eq:A-1-1}.
In this paper, we will see that the key integrable symplectic map constructed for solving \eqref{eq:1.1}
is definitely different from the one for solving \eqref{KN} used in \cite{XCZ-2022},
while we can also provide solutions for the discrete mKP equation \eqref{eq:A-1-1},
see Appendix \ref{A-1}.

The paper is organized as follows.
In Section \ref{sec-2}, we construct a nonlinear integrable symplectic map with the help of $r$-matrix.
In Section \ref{sec-3},  discrete flows are defined  by iterating the  symplectic map.
Resorting to the theory of Riemann surface,
Baker-Akhiezer functions are introduced and analyzed,
and the discrete potentials in terms of Riemann theta functions are obtained.
Then in Section \ref{sec-4}, we show that the obtained discrete potentials satisfy the coupled discrete CLL
system \eqref{eq:1.1} from the compatibility of the discrete phase flows.
Explicit formulae of algebro-geometric solutions are presented.
In the final section, we summarise our findings and mention some related problems.
There is an Appendix section, where solutions of the discrete mKP equation are derived.
Some obtained solutions of the discrete CLL system and the discrete mKP equation are illustrated.

\section{ Construction of the nonlinear integrable symplectic map}\label{sec-2}

To construct a nonlinear integrable symplectic map, we consider the following matrix
(Lax matrix)
\begin{equation}\label{eq:3.2}
L(\lambda;p,q)=\left(\begin{array}{cc}
L^{11}(\lambda) & L^{12}(\lambda) \\
L^{21}(\lambda) &-L^{11}(\lambda)
\end{array}\right)
=\left(\begin{array}{cc}
\frac{1+<p,q>}{2}+Q_\lambda(A^2p,q)&-\lambda Q_\lambda(Ap,p)\\
\lambda Q_\lambda(Aq,q) &-\frac{1+<p,q>}{2}-Q_\lambda(A^2p,q)
\end{array}\right),
\end{equation}
where $p=(p_{1},\ldots,p_{N})^{T}, q=(q_{1},\ldots,q_{N})^{T}$,
$A\doteq{\rm diag}(\alpha_1,\cdots,\alpha_N)$ subject to
$\alpha_1^2,\cdots,\alpha_N^2$ being distinct and non-zero constants,
and $Q_\lambda(\xi,\eta)=<(\lambda^2-A^2)^{-1}\xi,\eta>$ with $<\cdot,\cdot>$
being the inner product in $\mathbb{R}^{N}$.
For further analysis, we consider the determinant
\begin{equation}\label{eq:3.3}
    F(\lambda)=\mathrm{det}L(\lambda)=-\frac{(1+<p,q>)^2}{4} -(1+<p,q>)Q_{\lambda}(A^{2}p,q)
-Q_{\lambda}^2(A^2p,q)+\lambda^2Q(Ap,p)Q(Aq,q),
\end{equation}
where $L(\lambda)$ denotes $L(\lambda;p,q)$ for short.
For $\left|\lambda\right|>\mathrm{max}\left\{\left| \alpha_1 \right|,\cdots,\left| \alpha_N \right|\right\}$,
we get a power series expansion
\begin{subequations}\label{eq:3.4}
\begin{align}
    F(\lambda)&=\sum_{j=0}^{\infty}{F_k\zeta^{-k}},~~~ \zeta=\lambda^2,\\
    F_0&=-\frac{(1+<p,q>)^2}{4},\\
     F_k&=-(1+<p,q>)<A^{2k}p,q>-\sum_{i+j=k-2;i,j\geq0}<A^{2i+2}p,q><A^{2j+2}p,q> \nonumber \\
      &~~~~ +\sum_{i+j=k-1;i,j\geq 0}<A^{2i+1}p,p><A^{2j+1}q,q>,~~\ \ (k\geq1).
\end{align}
\end{subequations}
In the next, we prove that $F_0, F_1,\cdots, F_{N-1}$ are involutive by employing a $r$-matrix structure.
First, by some calculations, we obtain the $r$-matrices corresponding to the  matrix $L(\lambda)$,
which satisfy\footnote{
For $C=(C^{ij})_{2\times 2}$ and $D=(D^{ij})_{2\times 2}$ which are functions of
$p$ and $q$, the related  Poisson brackets are defined as the following:
\[\{C^{ij},D^{kl}\}=\sum^{N}_{s=1}\biggl(\frac{\partial C^{ij}}{\partial q_s}\frac{\partial D^{kl}}{\partial p_s}-
\frac{\partial C^{ij}}{\partial p_s}\frac{\partial D^{kl}}{\partial q_s}\biggr),
\]
\[ \{C^{ij},D\}=\left(\begin{array}{cc}
\{C^{ij}, D^{11}\}   & \{C^{ij}, D^{12}\} \\
\{C^{ij}, D^{21}\} & \{C^{ij}, D^{22}\}
\end{array}\right),\]
\[ \{C \otimes, D\}=\left(\begin{array}{cc}
\{C^{11}, D\}   & \{C^{12}, D \} \\
\{C^{21}, D \} & \{C^{22}, D \}
\end{array}\right).\]
}
\begin{subequations}
\begin{equation}\label{eq:3.7}
    \left\{L(\lambda)\otimes, L(\mu)\right\}
    =[r(\lambda,\mu),L(\lambda)\otimes I]+[r^{\prime}(\lambda,\mu),I\otimes L(\mu)],
\end{equation}
where
\begin{align}
&r(\lambda,\mu)=\frac{2\lambda}{\lambda^2-\mu^2}P_{\lambda\mu}-\frac{1}{2}(\sigma_3\otimes\sigma_3+I), 
\label{eq:3.7a}\\
&r^{\prime}(\lambda,\mu)
=\frac{2\mu}{\lambda^2-\mu^2}P_{\mu\lambda}+\frac{1}{2}(\sigma_3\otimes\sigma_3+I),\\
&P_{\lambda\mu}=\left( \begin{array}{cccc}
	\lambda&		0&		0&		0\\
	0&		0&		\mu&		0\\
	0&		\mu&		0&		0\\
	0&		0&		0&		\lambda\\
\end{array} \right) , ~\sigma_3=\left(\begin{array}{cc}
    1 &  0\\
   0  & -1
\end{array}\right),
\end{align}
\end{subequations}
$[A,B]=AB-BA$ and  $I$ generally stands for a unit matrix.
Note that we do not fix the size of $I$ without making confusion 
(e.g. $I$ is of $2\times 2$ in \eqref{eq:3.7} while $4\times 4$ in \eqref{eq:3.7a}).
Since $\mathrm{tr}(L(\lambda))=0$, we have  $L^{2}(\lambda)=-F(\lambda)I$ and $L^{2}(\mu)=-F(\mu)I$,
which lead us to
\begin{equation}
\{F(\lambda),F(\mu)\}=\displaystyle\frac{1}{4}\mathrm{tr}\{L^{2}(\lambda) \otimes, L^{2}(\mu)\}.
\end{equation}
In addition, similar to equation (3.22) in \cite{XCN-2021},
using \eqref{eq:3.7} we can express
 $\{L^{2}(\lambda) \otimes, L^{2}(\mu)\}$
 by a summation of two commutators.
It then follows that $\{F(\lambda),F(\mu)\}=0$.
Thus, from the expansion \eqref{eq:3.4} we know that
\begin{equation}
\label{FF}
\{F_{j},F_{k}\}=0,~~~~ ( j,k=0,1,2,\cdots).
\end{equation}

In the next step, we come to  prove that $F_0,\cdots,F_{N-1}$ are functionally independent.
To this end, we consider the factorization
\begin{equation}\label{eq:3.9}
    F(\lambda)=-\frac{(1+<p,q>)^2\prod_{j=1}^N(\zeta-\lambda_j^2)}{4\prod_{j=1}^N(\zeta-\alpha_j^2)}
    =-\frac{(1+<p,q>)^2R(\zeta)}{4\alpha^2(\zeta)},
\end{equation}
where an algebraic curve $(\zeta=\lambda^2)$ is involved:
\begin{equation}\label{eq:3.10}
\mathcal R:~~ \xi^2=R(\zeta)=\prod_{j=1}^N(\zeta-\alpha_j^2)(\zeta-\lambda_j^2).
\end{equation}
Note that here $\alpha^2(\zeta)=\prod_{j=1}^N(\zeta-\alpha_j^2)$.
Since $\mathrm{deg}R=2N$, the above curve $\mathcal R$ has two infinities $\infty_{+}$, $\infty_{-}$
and its genus is $g = N-1$. We denote the generic point on $\mathcal R$ as
\begin{equation}\label{point-R}
\mathfrak p(\zeta)=\big(\zeta,\xi=\sqrt{R(\zeta)}\big),\quad
(\tau\mathfrak p)(\zeta)=\big(\zeta,\xi=-\sqrt{R(\zeta)}\big),
\end{equation}
together with $\tau:\mathcal R\rightarrow\mathcal R$ being the hyperelliptic involution.
The associated basis of holomorphic differentials are
\begin{equation}\label{eq:3.11}
        \omega^{\prime}_s=\frac{\zeta^{g-s}\mathrm{d}\zeta}{2\sqrt{R(\zeta)}},~~~(s=1,\cdots,g ).
\end{equation}
Note that for the Hamiltonian function $F(\lambda)$, the corresponding canonical equations are \cite{Arnold}
\begin{equation}\label{3.13}
    \frac{\mathrm{d}}{\mathrm{d} t_{(\lambda)}}\left(\begin{array}{c}
         p_j \\
         q_j
    \end{array}\right)=\left(\begin{array}{c}
         -\partial F(\lambda)/\partial q_j \\
          \partial F(\lambda)/\partial p_j
    \end{array}\right)=W(\lambda,\alpha_j)\left(\begin{array}{c}
         p_j \\
         q_j
    \end{array}\right),\quad 1\leq j\leq N,
\end{equation}
where
\begin{equation}\label{eq:3.14}
    W(\lambda,\mu)
    =\frac{2\mu}{\lambda^2-\mu^2}\left(\begin{array}{cc}
      \mu L^{11}(\lambda)   &  \lambda L^{12}(\lambda)\\
       \lambda L^{21}(\lambda)  &  -\mu L^{11}(\lambda)
    \end{array}\right)+L_{11}(\lambda)\sigma_3,
\end{equation}
and $ t_{(\lambda)}$ denotes the corresponding flow variable.
One can verify that $L(\mu)$ and $W(\lambda,\mu)$ satisfy the Lax equation along the flow  corresponding to  $F(\lambda)$:
\begin{equation}\label{3.15}
    \frac{\mathrm{d}L(\mu)}{\mathrm{d}t_{(\lambda)}}=[W(\lambda,\mu),L(\mu)],\ \ \forall \lambda,\mu\in\mathbb{C}.
\end{equation}
Therefore the matrix $L(\lambda)$ defined in \eqref{eq:3.2} is also called a Lax matrix.
The Lax equation indicates a relation
\begin{equation}\label{eq:3.16}
    \frac{\mathrm{d}L^{12}(\mu)}{\mathrm{d}t_{(\lambda)}}
    =2(W^{11}(\lambda,\mu)L^{12}(\mu)-W^{12}(\lambda,\mu)L^{11}(\mu)).
\end{equation}
Meanwhile, from \eqref{eq:3.2} we can rewrite $L^{12}(\lambda)$ as
\begin{equation}\label{eq:3.12}
    L^{12}(\lambda)=-\lambda<Ap,p>\frac{s(\zeta)}{\alpha(\zeta)},
    ~~s(\zeta)=\prod_{j=1}^g(\zeta-\mu_j^2).
\end{equation}
Calculating equation  \eqref{eq:3.16} at the point $\mu=\mu_k$, we have
\begin{equation}\label{eq:3.17}
\frac{1}{2(1+<p,q>)\sqrt{R(\mu_k^2)}}\frac{\mathrm{d}(\mu_k^2)}{\mathrm{d}t_{\lambda}}
=\frac{\zeta}{\alpha(\zeta)}\frac{s(\zeta)}{(\zeta-\mu_k^2) s^{\prime}(\mu_k^2)}.
\end{equation}
Then, resorting to the interpolation formula of polynomial $s(\zeta)$, we arrive at
\begin{equation}\label{eq:3.18}
\sum_{k=1}^g \frac{(\mu_k^2)^{g-i}}{2(1+<p,q>)\sqrt{R(\mu_k^2)}}
\frac{\mathrm{d}(\mu_k^2)}{\mathrm{d}t_{\lambda}}
    =\frac{\zeta^{g-i+1}}{\alpha(\zeta)},~~~(i=1, 2,\cdots,g).
\end{equation}
Next, for the basis \eqref{eq:3.11}, i.e. $\{\omega^{\prime}_1, \omega^{\prime}_2, \cdots, \omega^{\prime}_g\}$,
we introduce the quasi Abel-Jacobi variables
\begin{equation}
\phi_i^{\prime}=\sum_{k=1}^g\int_{\mathfrak p_0}^{\mathfrak p(\mu_k^2)}\omega_i^{\prime},
~~~ (i=1,2,\cdots, g),
\end{equation}
where $\mathfrak p_0$ is a given point on $\mathcal{R}$
and $\mathfrak p(\mu_k^2)$ a point on $\mathcal{R}$ as well, with $\zeta=\mu_k^2$.
Then  we can get
\begin{equation}\label{eq:3.19}
 \{\phi_i^{\prime},F(\lambda)\}
 =\frac{(1+<p,q  >)\zeta^{g-i+1}}{\alpha(\zeta)}
 =\sum_{k=0}^{\infty}\{\phi_{i}^{\prime},F_k\}\zeta^{-k}
 =(1+<p,q>)\sum_{k=-\infty}^{\infty}A_{k-i} \zeta^{-k},
\end{equation}
with $A_0=1$ and $A_{-j}=0, \forall j\in\mathbb{Z}^+$.
Consequently, we arrive at
\begin{equation}\label{eq:3.21}
    \frac{\partial{(\phi_1^{\prime},\cdots,\phi_g^{\prime})}}{\partial{(t_1,\cdots t_g)}}
    =(\{\phi_i^{\prime},F_k\})_{g \times g}=(1+<p,q>)\left( \begin{array}{ccccc}
	1&		A_1&		A_2&		\cdots&		A_{g-1}\\
	&		1&		A_1&		\cdots&		A_{g-2}\\
	&		&		\ddots&		\ddots&		\vdots\\
	&		&		&		\ddots&		A_1\\
	&		&		&		&		1\\
\end{array} \right),
\end{equation}
where $t_k$ is the flow variable for the Hamiltonian function $F_k$.
We should note that $g=N-1$.
Now we are ready to present our result about the property of $F(\lambda)$.

\begin{lemma}\label{lem-2-1}
$F_0,F_1,\cdots,F_{N-1}$ are functionally independent in the open set $\left\{\mathrm{d}F_0\neq0\right\}$.
\end{lemma}

\begin{proof}
Assume there exist constants $\{c_k\}$ such that $\sum_{k=0}^{g}c_k\mathrm{d}F_k=0$,
and note that $g=N-1$.
Then we have
\begin{align}
 0=\omega^2\Bigl( I\sum_{k=0}^gc_k \mathrm{d} F_k, ~I\mathrm{d}\phi_i^{\prime}\Bigr)
 =\sum_{k=0}^gc_k\{{\phi_i^{\prime}},F_k\}=\sum_{k=1}^{N-1} c_k \{\phi_i^{\prime},F_k\},\label{eq:3.22}
\end{align}
where in the last step $g=N-1$ and  $(\phi_i^{\prime},F_0)=A_{-i}=0$ have been used.
This is an equation set with $N-1$ equations and $N-1$ unknown $\{c_k\}$.
In light of \eqref{eq:3.21} we can successively determine
$c_{N-1}=c_{N-2}=\cdots = c_{1}=0$.
Finally, only $c_0\mathrm{d}F_0=0$ is left,  which implies $c_0=0$ if $dF_0 \neq 0$.
Thus the proof is completed.

\end{proof}

Finally, we  introduce our integrable  map:
\begin{subequations}\label{eq:3.23}
\begin{align}
&S_\beta:\mathbb{R}^{2N}\rightarrow\mathbb{R}^{2N},\quad (p,q)\mapsto(\breve{p},\breve{q}),\label{eq:3.23a}\\
&{\breve{p}_j\choose \breve{q}_j}=\frac{1}{\sqrt{\beta^2-\alpha_j^2}}
D^{(\beta)}(\alpha_j){p_j\choose q_j},\quad 1\leq j\leq N\label{eq:3.23b},
\end{align}
where
\begin{equation}\label{eq:2.2}
D^{(\beta)}(\alpha_j)\doteq \frac{1}{\sqrt{\Upsilon}}\left(\begin{array}{cc}
                    -\alpha_j^2+\beta^2+a\breve{b}&\alpha_j a \\
                    \alpha_j \breve{b}&\beta^2
                   \end{array}\right),\quad \Upsilon=\beta^2+a\breve{b}.
\end{equation}
\end{subequations}
Note that this agrees with the spectral problem by some gauge transformations
(see Sec.4.2.2 in \cite{LZZ}) and deformations.

To make the map $S_\beta$ to be symplectic and integrable
\cite{Arnold,Goldstein,Jordan,Veselov,Veselov1,Bruschi,Suris},
we consider the following constraints
\begin{equation}\label{eq:3.25}
a=\frac{<Ap,p>}{1+<p,q>}, ~~\breve{b}=-\frac{<A\breve{q},\breve{q}>}{1+<p,q>},
\end{equation}
which indicates
\begin{subequations}\label{eq:3.24}
\begin{align}
&P_1(a)\equiv(1+<p,q>)a-<Ap,p>=0,\label{eq:3.24a}\\
&P_2 (\breve{b}/\beta ) \equiv
\frac{1}{\beta^3}\{\Upsilon[<A\breve{q},\breve{q}>+\breve{b}(1+<p,q>)]-\breve{b}^2P_1(a)\}=0. \label{eq:3.24b}
\end{align}
\end{subequations}
Among them, in light of \eqref{eq:3.23}, $P_2 (\breve{b}/\beta )$ can be converted into
\begin{equation}\label{eq:3.025}
P_2 (\breve{b}/\beta )= -\bigg( \frac{\breve{b}}{\beta} \bigg)^2L_{12}(\beta)
+2\frac{\breve{b}}{\beta}L_{11}(\beta)+L_{21}(\beta).
\end{equation}
Thus, the following explicit forms of the constraints are obtained:
\begin{subequations}\label{eq:3.26}
\begin{align}
&a=\frac{<Ap,p>}{1+<p,q>}, \label{eq:3.26b}\\
&b=\frac{1}{\beta^2 Q_\beta(Ap,p)}\Big(\frac{1+<p,q>}{2}+Q_\beta(A^2p,q)\pm \displaystyle\frac{(1+<p,q>)\sqrt{R(\beta^2)}}{2\alpha(\beta^2)}\Big), \label{eq:3.26c}
\end{align}
\end{subequations}
where $b$ is single-valued as a function of $\mathfrak{p}(\beta^{2})\in \mathcal{R}$.
For convenience we denote the above expression by
\begin{equation}\label{eq:3.25a}
(a,b)=f_{\beta}(p,q).
\end{equation}
Note that with these constraints, the map $S_{\beta}$ defined in \eqref{eq:3.23}
is no linger linear with respect to $(p,q)$.
In the following we show the map is symplectic.

\begin{proposition}\label{P-3-1}
The nonlinearized map $S_{\beta}$ is symplectic and Liouville integrable,
having the invariants $F_0,\cdots, F_{N-1}$ defined by \eqref{eq:3.4}.
\end{proposition}

\begin{proof}
By using equation (\ref{eq:3.23b}), the symplectic property is obtained by calculating
\begin{equation}\label{eq:3.27}
\sum_{j=1}^N ({\rm d}\breve{p}_j\wedge{\rm d}\breve{q}_j-{\rm d}p_j\wedge{\rm d}q_j)
 =-\frac{1}{2}{\rm d}\frac{P_1(a)}{\beta^2+a\breve{b}}\wedge{\rm d}\breve{b}
 +\frac{\beta}{2}{\rm d}a\wedge {\rm d}\frac{ P_2(\frac{\breve{b}}{\beta})}{\beta^2+a\breve{b}},
\end{equation}
which vanishes when constraint \eqref{eq:3.25} makes sense.
Besides, in light of the constraint (\ref{eq:3.26}), it turns out that
\begin{align}
&L(\lambda;\breve{p},\breve{q})D^{(\beta)}(\lambda;a,b)-D^{(\beta)}(\lambda;a,b)L(\lambda;p,q)) \nonumber \\
=&\frac{1}{2\Upsilon^{\frac{3}{2}}}\left(\begin{array}{cc} (\lambda^2+\Upsilon)\breve{b} &\lambda(2\beta^2+a\breve{b}) \\
  -\lambda \breve{b}^2 & -\beta^2\breve{b}
\end{array}\right)P_1(a) +
\frac{1}{2\Upsilon^{\frac{3}{2}}}\left(\begin{array}{cc}
  -(\lambda^2+\Upsilon)\beta a & \lambda\beta a^2 \\
  -\lambda\beta(2\beta^2+a\breve{b}) & \beta^3a
\end{array}\right)
P_2\bigg(\frac{\breve{b}}{\beta}\bigg)=0.
\label{eq:3.28}
\end{align}
Thus we arrive at $\det L(\lambda;\breve{p},\breve{q})=\det L(\lambda;p,q)$,
which indicates that
\begin{equation}\label{2.29}
F(\lambda;\breve{p},\breve{q})=F(\lambda;p,q)
\end{equation}
and then $F_0,\cdots,F_{N-1}$ are invariant under the action of
the map $S_{\beta}$.
The involutive property of $\{F_k\}$ has been proved in \eqref{FF}.

\end{proof}

\section{Evolution of the discrete potentials}\label{sec-3}

The functions $(a,b)$ expressed in \eqref{eq:3.26} will provide solutions to the discrete CLL system \eqref{eq:1.1}.
The purpose of this section is to introduce evolution to $(a,b)$
and reconstruct them in terms of the Riemann theta functions.

First, by iterating the nonlinear integrable symplectic map $S_{\beta}$ defined in \eqref{eq:3.23},
we get a discrete phase flow
\begin{equation}\label{3.1}
(p(m),\,q(m))=S_{\beta}(p(m-1),\,q(m-1)=\cdots =(S_{\beta})^m(p(0),q(0))=S_{\beta}^m(p(0),q(0))
\end{equation}
with any given initial point $(p(0),q(0))\in \mathbb{R}^{2N}$.
Along this $S_\beta^m$-flow, the potentials given by  \eqref{eq:3.26} (also see \eqref{eq:3.25a})  are expressed as
\begin{equation}\label{eq:4.1}
(a_m,b_m)=f_\beta\big(p(m),\,q(m)\big).
\end{equation}
And we get the commutative relation by equation \eqref{eq:3.28}
\begin{equation}\label{eq:4.4}
L_{m+1}(\lambda)D_m(\lambda)=D_m(\lambda)L_m(\lambda),
\end{equation}
where $L_m(\lambda)=L\big(\lambda;p(m),q(m)\big)$ is defined as the form \eqref{eq:3.2}
and
\begin{equation}\label{eq:4.5}
D_m(\lambda)=D^{(\beta)}(\lambda;a_m,b_m)
=\frac{1}{\sqrt{\Upsilon_m}}\left(\begin{array}{cc}
-\lambda^2+ \Upsilon_m&\lambda a_m\\
\lambda b_{m+1}&\beta^2
\end{array}\right),
\ \ \Upsilon_m=\beta^2+a_mb_{m+1}.
\end{equation}
In addition, from \eqref{eq:4.4} it is easy to find
\begin{subequations}\label{eq:4.6}
\begin{align}
&L_m(\lambda)M(m,\lambda)=M(m,\lambda)L_0(\lambda),\label{eq:4.6c}\\
&M(m,\lambda)=D_{m-1}(\lambda)D_{m-2}(\lambda)\cdots D_0(\lambda),\label{eq:4.6a}\\
&\det M(m,\lambda)=(\beta^2-\zeta)^m.\label{eq:4.6b}
\end{align}
\end{subequations}

Second, since $L^{2}(\lambda;p(m),q(m))=-F(\lambda;p(m),q(m))I$, it indicates that
$L_m(\lambda)$ has two distinct eigenvalues $\pm \rho_{\lambda}=\pm\sqrt{-F(\lambda)}$
which are independent of $m$ due to the property \eqref{2.29}.
We denote the two corresponding eigenvectors by
$h_{\pm}(m,\lambda)=(h_{\pm}^{(1)}(m,\lambda), h_{\pm}^{(2)}(m,\lambda))^T$.
For convenience in the following we normalize $h_{\pm}(0,\lambda)$
(satisfying $ L_0(\lambda)h_{\pm}(0,\lambda)=\pm \rho_{\lambda} h_{\pm}(0,\lambda)$)
to be
\begin{subequations}
\begin{equation}\label{eq:0a4.12}
h_{\pm}(0,\lambda)=(c_\lambda^\pm,1)^T,
\end{equation}
where
\begin{equation}\label{eq:4.12}
c_\lambda^\pm=\frac{L_0^{11}(\lambda)\pm
\rho_{\lambda}}{L_0^{21}(\lambda)}
=\frac{-L_0^{12}(\lambda)}{L_0^{11}(\lambda)\mp \rho_{\lambda}}.
\end{equation}
\end{subequations}
Thus, from \eqref{eq:4.6c} and \eqref{eq:0a4.12} we may write the eigenvectors as
\begin{equation}
h_{\pm}(m,\lambda)=M(m,\lambda){c_\lambda^\pm\choose 1},\label{eq:4.10}
\end{equation}
by letting $M(0,\lambda)$ being the unit matrix $I$. In addition, from \eqref{eq:4.4} we have
\begin{equation}
h_{\pm}(m+1,\lambda)=D_m(\lambda)h_{\pm}(m,\lambda).\label{eq:4.9b}
\end{equation}

In the third part of this section, we will study these eigenvectors $h_{\pm}(m,\lambda)$
under  the theory of Riemann surface \cite{Farkas,Griffiths,Mumford}.
Let us introduce four meromorphic functions
(Baker-Akhiezer functions) on $\mathcal R$,
\begin{equation}\label{eq:4.13}
\begin{split}
&\mathfrak{h}^{(1)}\big(m,\mathfrak p(\lambda^2)\big)=\lambda h_+^{(1)}(m,\lambda),\quad
\mathfrak{h}^{(1)}\big(m, \tau \mathfrak p(\lambda^2)\big)=\lambda h_-^{(1)}(m,\lambda),\\
&\mathfrak{h}^{(2)}\big(m,\mathfrak p(\lambda^2)\big)=h_+^{(2)}(m,\lambda),\quad
\mathfrak{h}^{(2)}\big(m, \tau\mathfrak  p(\lambda^2)\big)=h_-^{(2)}(m,\lambda),
\end{split}\end{equation}
 and the elliptic variables $\mu^{2}_j,\,\nu^{2}_j$ defined by
\begin{subequations}\label{eq:4.14}
\begin{align}
&\lambda^{-1}L_m^{12}(\lambda)=-Q_\lambda(Ap(m),p(m))
=-\frac{(1+<p(m),q(m)>)a_m}{\alpha(\zeta)}\prod_{j=1}^{N-1}
\big(\zeta-\mu_j^2(m)\big),\\
&\lambda^{-1}L_m^{21}(\lambda)=Q_\lambda(Aq(m),q(m))
=-\frac{(1+<p(m),q(m)>)b_m}{\alpha(\zeta)}\prod_{j=1}^{N-1}
\big(\zeta-\nu_j^2(m)\big).
\end{align}
\end{subequations}
As a consequence, by using equations (\ref{eq:4.6}) and (\ref{eq:4.10}) we obtain
\begin{subequations}\label{eq:4.15}
\begin{align}
&\mathfrak{h}^{(1)}\big(m,\mathfrak p(\lambda^2)\big)\cdot
\mathfrak{h}^{(1)}\big(m, \tau\mathfrak p(\lambda^2)\big)
=-\zeta(\beta^2-\zeta)^m\frac{a_m}{b_0}\prod_{j=1}^{N-1}\frac{\zeta-\mu_j^2(m)}{\zeta-\nu_j^2(0)},
\label{eq:4.15a}\\
&\mathfrak{h}^{(2)}\big(m,\mathfrak p(\lambda^2)\big)\cdot
\mathfrak{h}^{(2)}\big(m, \tau\mathfrak p(\lambda^2)\big)
=(\beta^2-\zeta)^m\frac{b_m}{b_0}\prod_{j=1}^{N-1}\frac{\zeta-\nu_j^2(m)}{\zeta-\nu_j^2(0)}.\label{eq:4.15b}
\end{align}
\end{subequations}

The asymptotic behaviors of the above meromorphic functions are described as the following.

\begin{lemma}\label{L-4.2}
The following asymptotic behaviors hold when  $\zeta=\lambda^2\sim\infty$:
\begin{subequations}\label{eq:4.16}
\begin{align}
&\mathfrak{h}^{(1)}\big(m,\mathfrak p(\lambda^2)\big)
=\frac{(-\zeta)^{m+1}}{b_0\sqrt{\Upsilon_0\Upsilon_1\cdots \Upsilon_{m-1}}}\big(1+O(\zeta^{-1})\big),\label{eq:4.16a}\\
&\mathfrak{h}^{(1)}\big(m, \tau\mathfrak p(\lambda^2)\big)
=a_m\sqrt{\Upsilon_0\Upsilon_1\cdots \Upsilon_{m-1}}\big(1+O(\zeta^{-1})\big),\label{eq:4.16b}\\
&\mathfrak{h}^{(2)}\big(m,\mathfrak p(\lambda^2)\big)
=\frac{b_m(-\zeta)^m}{b_0\sqrt{\Upsilon_0\Upsilon_1\cdots \Upsilon_{m-1}}}\big(1+O(\zeta^{-1})\big),\label{eq:4.16c}\\
&\mathfrak{h}^{(2)}\big(m, \tau\mathfrak p(\lambda^2)\big)
=\sqrt{\Upsilon_0\Upsilon_1\cdots \Upsilon_{m-1}}\big(1+O(\zeta^{-1})\big).\label{eq:4.16d}
\end{align}\end{subequations}
\end{lemma}

\begin{proof}
First, from the expression \eqref{eq:4.6a}, we have (for $ m\geq2$)
\begin{subequations}
\begin{align}\label{4.7}
&M^{11}(m,\lambda)=\frac{(-\zeta)^m}{\sqrt{\Upsilon_0\Upsilon_1\cdots \Upsilon_{m-1}}}\big(1+O(\zeta^{-1})\big),\\
&\lambda M^{12}(m,\lambda)=\frac{-a_0(-\zeta)^m}{\sqrt{\Upsilon_0\Upsilon_1\cdots \Upsilon_{m-1}}}\big(1+O(\zeta^{-1})\big),\\
&\lambda M^{21}(m,\lambda)=\frac{-b_m(-\zeta)^m}{\sqrt{\Upsilon_0\Upsilon_1\cdots \Upsilon_{m-1}}}\big(1+O(\zeta^{-1})\big),\\
&M^{22}(m,\lambda)=\frac{-a_0b_m(-\zeta)^{m-1}}{\sqrt{\Upsilon_0\Upsilon_1\cdots \Upsilon_{m-1}}}\big(1+O(\zeta^{-1})\big).
\end{align}
\end{subequations}
The first three relations are still valid for $m=1$,
while $M^{22}(1,\lambda)$ is a constant when $m=1$: $M^{22}(1,\lambda)=\beta^2/\sqrt{\Upsilon_{0}}$.
In addition, by using equations (\ref{eq:4.12}) and (\ref{eq:3.2}) we get
\[\lambda c_\lambda^+=(-\zeta/b_0)\big(1+O(\zeta^{-1})\big),~~~ ~(\zeta \sim\infty).\]
Then, from \eqref{eq:4.10} and \eqref{eq:4.13}, we find (\ref{eq:4.16a}) and (\ref{eq:4.16c}).
The other two in \eqref{eq:4.16} can be derived  from \eqref{eq:4.15} by using  (\ref{eq:4.16a}) and (\ref{eq:4.16c}).

\end{proof}

\begin{lemma}\label{L-4.3}
When $\zeta=\lambda^2\sim 0$, we have the following asymptotic behaviors:
\begin{subequations}\label{eq:4.17}
\begin{align}
&\mathfrak{h}^{(1)}\big(m,\mathfrak p(\lambda^2)\big)
=-\frac{1-<p(0),q(0)>}{<A^{-1}q(0),q(0)>}\sqrt{\Upsilon_0\Upsilon_1\cdots \Upsilon_{m-1}}\big(1+O(\zeta)\big),\\
&\mathfrak{h}^{(1)}\big(m, \tau\mathfrak p(\lambda^2)\big)
=-\frac{\zeta\beta^{2m}<A^{-1}p(m),p(m)>
}{(1-<p(0),q(0)>)\sqrt{\Upsilon_0\Upsilon_1\cdots \Upsilon_{m-1}}}\big(1+O(\zeta)\big).
\end{align}
\end{subequations}
\end{lemma}

\begin{proof}
The proof is similar as for Lemma \ref{L-4.2}.
First, equation (\ref{eq:4.12}) yields
$$\lambda c_\lambda^+=-\frac{1-<p(0),q(0)>}{<A^{-1}q(0),q(0)>}\big(1+O(\zeta)\big),~~~~(\lambda\sim 0).$$
In addition, when $\zeta\sim 0$, from \eqref{eq:4.6a} we have
\begin{subequations}\label{eq:4.8}
\begin{align}
&M^{11}(m,\lambda)=\sqrt{\Upsilon_0\Upsilon_1\cdots \Upsilon_{m-1}}\big(1+O(\zeta) \big) (m\geq 1),
\label{eq:4.8a} \\
& \lambda M^{12}(m,\lambda)=O(\zeta)((m\geq 1)),\quad\lambda M^{21}(m,\lambda)=O(\zeta)(m\geq 1),
\label{eq:4.8b}\\
&M^{22}(m,\lambda)=\frac{\beta^{2m}}{\sqrt{\Upsilon_0\Upsilon_1\cdots
\Upsilon_{m-1}}}\big(1+O(\zeta)\big)(m\geq 2),\quad  M^{22}(1,\lambda)=\beta^2/\sqrt{\Upsilon_{0}}.
\label{eq:0.8b}
    \end{align}
\end{subequations}
Then, the two relations in (\ref{eq:4.17}) are the consequence of (\ref{eq:4.10}), (\ref{eq:4.13}) and \eqref{eq:4.15}.

\end{proof}

With the asymptotic behaviors in hand, we can have the divisors of
$\mathfrak h^{(1)}(m,\mathfrak p),\,\mathfrak h^{(2)}(m,\mathfrak p)$ on $\mathcal{R}$
\cite{Farkas,Griffiths,Mumford}:
\begin{subequations}\label{eq:4.18}
\begin{align}
&\mathcal{D}(\mathfrak h^{(1)}(m,\mathfrak p))
=\sum_{j=1}^g\Big(\mathfrak p\big(\mu_j^2(m)\big)-\mathfrak p\big(\nu_j^2(0)\big)\Big)
+\{\mathfrak o_-\}+m\{\mathfrak p(\beta^2)\}-(m+1)\{\infty_+\},\label{eq:4.17a}\\
&\mathcal{D}(\mathfrak h^{(2)}(m,\mathfrak p))=\sum_{j=1}^g\Big(\mathfrak p\big(\nu_j^2(m)\big)
-\mathfrak p\big(\nu_j^2(0)\big)\Big)+m\{\mathfrak p(\beta^2)\}-m\{\infty_+\},\label{eq:4.17b}
\end{align}\end{subequations}
where $\mathfrak o_-=(\zeta=0,\,\xi=-\sqrt{R(0)})$, $ g=N-1$.

Next, to reconstruct $\mathfrak h^{(1)}(m,\mathfrak p),\,\mathfrak h^{(2)}(m,\mathfrak p)$,
let us introduce the Abel-Jacobi variables
\begin{equation}\label{eq:4.19}
\vec{\psi}(m)=\mathcal{A}\Big(\hbox{$\sum$}_{j=1}^g\mathfrak p\big(\mu_j^2(m)\big)\Big), \quad
\vec{\phi}(m)=\mathcal{A}\Big(\hbox{$\sum$}_{j=1}^g\mathfrak p\big(\nu_j^2(m)\big)\Big)
\end{equation}
by using the Abel map $\mathcal A:\mathcal R\rightarrow J(\mathcal R)$ defined as
\begin{equation}\label{eq:0.5}
\mathcal{A}(\mathfrak p)=\int_{\mathfrak{p}_0}^{\mathfrak{p}}\vec{\omega}
=\left(\int_{\mathfrak{p}_0}^{\mathfrak{p}}\omega_1, \int_{\mathfrak{p}_0}^{\mathfrak{p}}\omega_2,
~\cdots,\int_{\mathfrak{p}_0}^{\mathfrak{p}}\omega_g\right)^T,
\end{equation}
with $\mathfrak{p}_0$ being any given point on $\mathcal{R}$.
For further discussions we introduce the canonical basis $a_{1},\cdots,a_{g},b_{1},\cdots,b_{g}$
for the homology group $\mathrm{H}_{1}(\mathcal{R})$.
Then $(\omega^{\prime}_1,\cdots,\omega^{\prime}_g)^{T}$ given by \eqref{eq:3.11} can be normalized as
\begin{subequations}\label{eq:0.0}
\begin{equation}
 \vec{\omega}\doteq (\omega_1,\cdots,\omega_g)^T
=C(\omega^{\prime}_{1},\cdots,\omega^{\prime}_g)^T
\end{equation}
where
\begin{equation}
C=(A_{jk})^{-1}_{g\times g},\ \ A_{jk}=\int_{a_{k}}\omega^{\prime}_{j}.
\end{equation}
\end{subequations}
And we have
\begin{equation}\label{eq:0.1}
\int_{a_{k}}\vec{\omega}=\vec{\delta}_{k},\ \ \int_{b_{k}}\vec{\omega}=\vec{B}_{k},
\end{equation}
which give rise to the period lattice
\[\mathcal T=\{z\in \mathbb{C}^{g}\mid z=\vec{s}+B\vec{r},~~ \vec{s},\vec{r}\in \mathbb{Z}^{g}\},\]
where $\vec{\delta}_{k}$ is the \emph{k}-th column of the usual unit matrix, and $B=(\vec{B}_{1},\vec{B}_{2},\cdots,\vec{B}_{g})=(B_{ij})_{g\times g}$
is symmetric with positive definite imaginary part and serves as a periodic matrix in the Riemann theta function.
Resorting to Toda's dipole technique \cite{Toda}, we obtain
\begin{subequations}\label{eq:4.20}
\begin{align}
&\vec\psi(m)\equiv\vec\phi(0)+m\vec\Omega_\beta+\vec\Omega_0,\quad({\rm mod}\,\mathcal T),\\
&\vec\phi(m)\equiv\vec\phi(0)+m\vec\Omega_\beta,\quad({\rm mod}\,\mathcal T),\\
&\vec\Omega_\beta=\int_{\mathfrak p(\beta^2)}^{\infty_+}\vec\omega,\quad
\vec\Omega_0=\int_{\mathfrak o_-}^{\infty_+}\vec\omega,
\end{align}
\end{subequations}
which indicate that discrete flow $S^{m}_{\beta}$ is straightened out in the Jacobi variety 
$J(\mathcal {R})=\mathbb C^g/\mathcal T$. As a result, by comparing the devisors,
we obtain the Riemann theta function expressions
for $\mathfrak h^{(1)}(m,\mathfrak p)$ and $\mathfrak h^{(2)}(m,\mathfrak p)$ \cite{CaoX-JPA-2012,Farkas,Griffiths,Mumford,XCZ-2022},
\begin{subequations}\label{eq:4.21}
\begin{align}
&\mathfrak h^{(1)}(m,\mathfrak p)=C_m^{(1)}
\frac{\theta(-\mathcal A(\mathfrak p)+\vec\psi(m)+\vec K;B)}
{\theta(-\mathcal A(\mathfrak p)+\vec\phi(0)+\vec K;B)}
\exp\int_{\mathfrak p^{\prime}_0}^{\mathfrak p}(m\,\omega[\mathfrak p(\beta^2),\infty_+]
+\omega[\mathfrak o_-,\infty_+]),\label{eq:4.21a}\\
&\mathfrak h^{(2)}(m,\mathfrak p)=C_m^{(2)}
\frac{\theta(-\mathcal A(\mathfrak p) +\vec\phi(m)+\vec K;B)}
{\theta(-\mathcal A(\mathfrak p)+\vec\phi(0)+\vec K;B)}
\exp\int_{\mathfrak p^{\prime}_0}^{\mathfrak p}m\,\omega[\mathfrak p(\beta^2),\infty_+],\label{eq:4.21b}
\end{align}
\end{subequations}
where $C_m^{(1)}$ and $C_m^{(2)}$ are constant factors, $\vec K$ is the Riemann constant vector,
$\mathfrak p^{\prime}_0$ is any given point on $\mathcal{R}$,
$\omega[\mathfrak p,\mathfrak q]$ is a meromorphical differential that has only simple
poles at $\mathfrak p$ and $\mathfrak q$ with residues $+1$ and $-1$,  respectively.
Here $\theta(z;B)$ is the Riemann theta function
(which is holomorphic)
\begin{equation}\label{eq:0.2}
\theta(z;B)=\sum_{z^{\prime}\in \mathbb{Z}^{g}}\exp\pi\sqrt{-1}(<Bz^{\prime},z^{\prime}>+2<z,z^{\prime}>),
\ \ z\in \mathbb{C}^{g},
\end{equation}
with $<\cdot,\cdot>$ being the scalar product in $\mathbb{C}^{g}$.

In order to obtain explicit expressions of discrete potentials $(a_m,b_m)$
in terms of the Riemann theta functions,
taking $\mathfrak p$ going to $\pm\infty$ in the Baker-Akhiezer functions given in \eqref{eq:4.21}
and then comparing the results with \eqref{eq:4.16},
we can find  the following relations:
\begin{subequations}\label{eq:4.28}
    \begin{align}
          &\frac{(-1)^{m+1}}{b_0\sqrt{\Upsilon_0\Upsilon_1\cdots \Upsilon_{m-1}}}=C_m^{(1)}
\frac{\theta(-\mathcal A(\infty_+)+\vec\psi(m)+\vec K;B)}
{\theta(-\mathcal A(\infty_+)+\vec\phi(0)+\vec K;B)}(\gamma_+^{\beta^2})^m\gamma_+,\label{eq:4.28a}\\
&a_m\sqrt{\Upsilon_0\Upsilon_1\cdots \Upsilon_{m-1}}=C_m^{(1)}
\frac{\theta(-\mathcal A(\infty_-)+\vec\psi(m)+\vec K;B)}
{\theta(-\mathcal A(\infty_-)+\vec\phi(0)+\vec K;B)}\exp\int_{\mathfrak p^{\prime}_0}^{\infty_-}m\,\omega[\mathfrak p(\beta^2),\infty_+]+\omega[\mathfrak o_-, \infty+],\label{eq:4.28b}\\
&\frac{b_m (-1)^m}{b_0\sqrt{\Upsilon_0\Upsilon_1\cdots \Upsilon_{m-1}}}=C_m^{(2)}
\frac{\theta(-\mathcal A(\infty_+)+\vec\phi(m)+\vec K;B)}
{\theta(-\mathcal A(\infty_+)+\vec\phi(0)+\vec K;B)}(\gamma_+^{(\beta^2)})^m,\label{eq:4.28c}\\
&\sqrt{\Upsilon_0\Upsilon_1\cdots \Upsilon_{m-1}}=C_m^{(2)}
\frac{\theta(-\mathcal A(\infty_-)+\vec\phi(m)+\vec K;B)}
{\theta(-\mathcal A(\infty_-)+\vec\phi(0)+\vec K;B)}\exp\int_{\mathfrak p^{\prime}_0}^{\infty_-}m\,\omega[\mathfrak p(\beta^2),\infty_+],\label{eq:4.28d}
\end{align}
\end{subequations}
where
\begin{equation}\label{eq:4.27}
\begin{aligned}
    &\gamma_+^{\beta^2}=\lim_{\mathfrak p\to \infty_+}\frac{1}{\zeta}
    \exp\int_{\mathfrak p^{\prime}_0}^{\mathfrak p} \omega[\mathfrak p(\beta^2),\infty_+],\\
    &\gamma_+=\lim_{\mathfrak p\to \infty_+}\frac{1}{\zeta}
    \exp\int_{\mathfrak p^{\prime}_0}^{\mathfrak p} \omega[\mathfrak o_-,\infty_+].
\end{aligned}
\end{equation}
By cancelling $\sqrt{\Upsilon_0\Upsilon_1\cdots \Upsilon_{m-1}}$ in \eqref{eq:4.28b} and \eqref{eq:4.28c}, we get
\begin{subequations}\label{eq:4.29}
    \begin{align}
       &\frac{\sqrt{\Upsilon_m}a_{m+1}}{a_m}=\frac{C_{m+1}^{(1)}}{C_m^{(1)}}\frac{\theta(-\mathcal A(\infty_-)+\vec\psi(m+1)+\vec K;B)}
{\theta(-\mathcal A(\infty_-)+\vec\psi(m)+\vec K;B)}
\exp\int_{\mathfrak p^{\prime}_0}^{\infty_-}\omega[\mathfrak p(\beta^2),\infty_+],\label{eq:4.29a}\\
&\frac{b_{m+1}}{\sqrt{\Upsilon_m}b_m}
=-\frac{C_{m+1}^{(2)}}{C_m^{(2)}}\frac{\theta(-\mathcal A(\infty_+)+\vec\phi(m+1)+\vec K;B)}
{\theta(-\mathcal A(\infty_+)+\vec\phi(m)+\vec K;B)}\gamma_+^{\beta^2}.\label{eq:4.29b}
   \end{align}
\end{subequations}
Now we consider the first row in equation \eqref{eq:4.9b}, i.e.
\begin{equation}\label{eq:4.30}
   \mathfrak h^{(1)}(m+1,\mathfrak p)=\frac{\lambda^2}{\sqrt{\Upsilon_m}}
   (a_m\mathfrak h^{(2)}(m,\mathfrak p) -\mathfrak h^{(1)}(m,\mathfrak p))+\sqrt{\Upsilon_m}
   \mathfrak h^{(1)}(m,\mathfrak p),
\end{equation}
which reads
\begin{equation}\label{eq:4.31}
   \mathfrak h^{(1)}(m+1,\mathfrak 0_+)=\sqrt{\Upsilon_m}\mathfrak h^{(1)}(m,\mathfrak 0_+)
\end{equation}
at the point $\mathfrak 0_+=(\zeta=0,\,\xi=\sqrt{R(0)})$.
Then, substituting \eqref{eq:4.21a} into \eqref{eq:4.31} we have
\begin{equation}\label{eq:4.32}
   \frac{C_{m+1}^{(1)}}{C_m^{(1)}}=
\sqrt{\Upsilon_m}\frac{\theta[\vec\psi(m)+\vec K +\vec\eta_{\mathfrak o+};B]}
{\theta[\vec\psi(m+1)+\vec K+\vec\eta_{\mathfrak o_+};B]}
\exp\int_{\mathfrak o_+}^{\mathfrak p^{\prime}_0}\omega[\mathfrak p(\beta^2),\infty_+],
\end{equation}
where $\vec{\eta}_{\mathfrak o_+}=-\mathcal A(\mathfrak o_+)$.
On the other hand, from the second row in equation \eqref{eq:4.9b} we obtain
\begin{equation}\label{eq:4.33}
   \frac{C^{(2)}_{m+1}}{C^{(2)}_{m}} =\frac{\beta^2}{\sqrt{\Upsilon_m}}
   \frac{\theta(\vec\phi(m)+\vec K+\vec{\eta}_{\mathfrak0_-};B)}
 {\theta(\vec\phi(m+1)+\vec K+\vec{\eta}_{\mathfrak0_-};B)}
 \exp\int_{\mathfrak0_-}^{\mathfrak p^{\prime}_0}\omega[\mathfrak p(\beta^2),\infty_+],
 \end{equation}
with $\vec\eta_{\mathfrak0_-}=-\mathcal A(\mathfrak0_-)$  since $\mathfrak h^{(1)}(m,\mathfrak o_-)=0$.
Next, substituting equations \eqref{eq:4.32} and \eqref{eq:4.33} into \eqref{eq:4.29} we arrive at
\begin{subequations}\label{eq:4.34}
     \begin{align}
        &a_m=a_0\frac{\theta(\vec\psi(m)+\vec K+\vec{\eta}_{\infty_-};B)
        \theta(\vec\psi(0)+\vec K+\vec{\eta}_{\mathfrak{o_+}};B)}
{\theta(\vec\psi(0)+\vec K+\vec{\eta}_{\infty_-};B)
\theta(\vec\psi(m)+\vec K+\vec{\eta}_{\mathfrak{o}_+};B)}
\exp\int_{\mathfrak o_+}^{\infty_-}m\omega[\mathfrak p(\beta^2),\infty_+],\label{eq:4.34a}\\
 &b_m=b_0(-\beta^2)^m\frac{\theta(\vec\phi(m)+\vec K+\vec{\eta}_{\infty_+};B)
 \theta(\vec\phi(0)+\vec K+\vec{\eta}_{\mathfrak o_-};B)}
{\theta(\vec\phi(0)+\vec K+\vec{\eta}_{\infty_+};B)
\theta(\vec\phi(m)+\vec K+\vec{\eta}_{\mathfrak o_-};B)}(\gamma_+^{\beta^2})^m
\exp\int_{\mathfrak o_-}^{\mathfrak p^{\prime}_0}m\omega[\mathfrak p(\beta^2),\infty_+],\label{eq:4.34b}
    \end{align}
\end{subequations}
where $\vec\eta_{\infty_+}=-\mathcal A(\infty_+), \vec\eta_{\infty_-}=-\mathcal A(\infty_-)$.
The expression of $\Upsilon_m$ is calculated by using \eqref{eq:4.28a} and \eqref{eq:4.32}, i.e.
\begin{equation}\label{eq:4.35}
    \Upsilon_m=-\frac{\theta(\vec\psi(m)+\vec K+\vec{\eta}_{\infty_+};B)
    \theta(\vec\psi(m+1)+\vec K+\vec{\eta}_{\mathfrak o_+};B)}
{\theta(\vec\psi(m+1)+\vec K+\vec{\eta}_{\infty_+};B)
\theta(\vec\psi(m)+\vec K+\vec{\eta}_{\mathfrak{o_+}};B)}
\cdot \frac{\exp\int_{\mathfrak p^{\prime}_0}^{\mathfrak{o_+}}
\omega[\mathfrak p(\beta^2),\infty_+]}{\gamma_+^{\beta^2}}.
\end{equation}

Thus, we have reconstructed the finite genus potentials $(a_m,b_m)$  using the Riemann theta functions,
which will lead to the algebro-geometric solutions for the discrete CLL system \eqref{eq:1.1}.

\section{Algebro-geometric solutions to the discrete CLL equation}\label{sec-4}

\subsection{Compatibility of the symplectic maps and the discrete CLL system}\label{sec-4-1}

In Section \ref{sec-2} we have constructed an integrable symplctic map $S_{\beta}$ in (\ref{eq:3.23})
by which we defined a discrete phase flow in \eqref{3.1}.
Now taking $(\beta,m)=(\beta_1,m_1)$ and $(\beta,m)=(\beta_2,m_2)$ respectively,  we can have two
integrable symplctic maps $S_{\beta_{1}}$ and $S_{\beta_{2}}$.
They  share  same integrals $\{F_{0},\cdots,F_{N-1}\}$ and the spectral curve $\mathcal {R}$.
According to the discrete analogue of Liouville-Arnold theorem \cite{Arnold,Goldstein,Jordan,Veselov,Veselov1,Bruschi,Suris},
the corresponding discrete phase flows $S^{m_{1}}_{\beta_{1}}$ and $S_{\beta_{2}}^{m_{2}}$ commute.

In the following, let us show that the discrete CLL system \eqref{eq:1.1} is a consequence
of the compatibility of $(\beta,m)=(\beta_1,m_1)$ and $(\beta,m)=(\beta_2,m_2)$.
Equivalently, we consider two linear problems:
\begin{subequations}\label{eq:2.6}
\begin{align}
&\t{\chi}=D^{(\beta_1)}\chi=\frac{1}{\sqrt{\Upsilon^{(1)}}}\left(\begin{array}{cc}
        -\lambda^2+\Upsilon^{(1)}&\lambda a \\
        \lambda \t{b}&\beta^2_1
        \end{array}\right)\chi,\ \ \Upsilon^{(1)}=\beta_{1}^2+a\t{b},\label{eq:2.6a} \\
&\overline{\chi}=D^{(\beta_2)}\chi=\frac{1}{\sqrt{\Upsilon^{(2)}}}\left(\begin{array}{cc}
        -\lambda^2+\Upsilon^{(2)}&\lambda a \\
        \lambda \overline{b}&\beta^2_2
        \end{array}\right)\chi,\ \ \Upsilon^{(2)}=\beta_{2}^2+a\overline{b}.\label{eq:2.6b}
\end{align}
\end{subequations}
Note that $D^{(\beta_i)}$ is defined as in \eqref{eq:4.5}.
Their compatibility gives rise to
\begin{align}
&  \t{D}^{(\beta_2)}D^{(\beta_1)}-\overline{D}^{(\beta_1)}D^{(\beta_2)} \nonumber\\
= ~   &\frac{1}{\sqrt{\Upsilon^{(1)} \t{\Upsilon}^{(2)}}}
\left(\begin{array}{cc}
\lambda^2Z_{12}+\Upsilon^{(1)}\t{\Upsilon}^{(2)}-\overline{\Upsilon}^{(1)} \Upsilon^{(2)}& -\lambda\Xi_1 \\
-\lambda \Xi_2 & 0
\end{array}\right) \nonumber \\
    &+\frac{\overline{\Upsilon}^{(1)}\Upsilon^{(2)}-\Upsilon^{(1)}\t{\Upsilon}^{(2)}}
    {(\sqrt{\overline{\Upsilon}^{(1)}
     \Upsilon^{(2)}}+\sqrt{\Upsilon^{(1)}\t{\Upsilon}^{(2)}})
     \sqrt{\Upsilon^{(1)}\t{\Upsilon}^{(2)}\overline{\Upsilon}^{(1)} \Upsilon^{(2)}}} \nonumber\\
    &\times \left(\begin{array}{cc}
        \lambda^4-\lambda^2(\Upsilon^{(2)}+\overline{\Upsilon}^{(1)}-\overline{a} \overline{b})
        +\overline{\Upsilon}^{(1)}
        \Upsilon^{(2)} 
        & -\lambda^3a+\lambda(a\overline{\Upsilon}^{(1)}+\beta_2^2\overline{a})\\
        \lambda^3\overline{\t{b}}+\lambda(\overline{\t{b}}\Upsilon^{(2)}+\beta_1^2\overline{b})
         & \lambda^2a\overline{\t{b}}+\beta_1^2\beta_2^2
    \end{array}\right).\label{eq:1.02}
\end{align}
where $\Xi_1$ and $\Xi_2$ are  given as in \eqref{eq:1.1},
\begin{equation*}
  \begin{aligned}
     Z_{12}&=\overline{\Upsilon}^{(1)}+{\Upsilon}^{(2)}-\t{{\Upsilon}}^{(2)}
     -{\Upsilon}^{(1)}+\t{a}\t{b}-\overline{a}\overline{b}\\
           &=\frac{1}{a}\Xi_1+\frac{1}{\overline{\t{b}}}\Xi_2
           +\frac{1}{a\overline{\t{b}}}(\Upsilon^{(1)}\t{\Upsilon}^{(2)}-\overline{\Upsilon}^{(1)} \Upsilon^{(2)}),
  \end{aligned}
\end{equation*}
and
\begin{equation*}
    \Upsilon^{(1)}\t{\Upsilon}^{(2)}-\overline{\Upsilon}^{(1)} \Upsilon^{(2)}
    =-\overline{a}\Xi_2-\t{b}\Xi_1+\frac{\beta_1^2}{\beta_1^2-\beta_{2}^2}
    [(\overline{a}-\t{a})\Xi_2-(\overline{b}-\t{b})\Xi_1].
\end{equation*}
This immediately yields the discrete CLL system \eqref{eq:1.1}.

\subsection{Algebro-geometric solutions of the discrete CLL system}\label{sec-4-1}

We define $\big(p(m_{1},m_{2}),q(m_{1},m_{2})\big)=S^{m_{1}}_{\beta_{1}}S_{\beta_{2}}^{m_{2}}(p_{0},q_{0})$,
and in the Jacobi variety $J(\mathcal {R})$ we get
\begin{equation}\label{eq:0.01}
\begin{split}
&\vec{\phi}(m_{1},m_{2})= \vec{\phi}(0,0)+m_{1}\vec{\Omega}_{\beta_{1}} +m_{2}\vec{\Omega}_{\beta_{2}},\\
&\vec{\psi}(m_{1},m_{2})= \vec{\phi}(0,0)+m_{1}\vec{\Omega}_{\beta_{1}} +m_{2}\vec{\Omega}_{\beta_{2}}+\vec{\Omega}_{0},\\
&\vec{\Omega}_{0}=\int_{\mathfrak o_-}^{\infty_+}\vec\omega,\quad 
\vec{\Omega}_{\beta_{k}}=\int_{\mathfrak p(\beta_{k}^2)}^{\infty_+}\vec\omega,\ \ k=1,2.
\end{split}
\end{equation}
Note that when $\lambda=\alpha_j$, the components $(p_j(m_{1},m_{2}),q_j(m_{1},m_{2}))$
satisfy the equation (\ref{eq:3.23}) for $\beta=\beta_1,\beta_2$ simultaneously,
which means
the potentials $a(m_1,m_2),b(m_{1},m_{2})$ solve the discrete CLL system.
Then we immediately have the following.

\begin{theorem}\label{T-4}
The discrete CLL system (\ref{eq:1.1}) admit the following algebro-geometric solutions:
\begin{subequations}\label{eq:4.37ab}
\begin{align}
a(m_1,m_2)=&a(0,0)\frac{\theta(\sum_{k=1}^2m_k\vec\Omega_{\beta_k}
+\vec{\Omega}_{0}+\vec\phi(0,0)
+\vec K+\vec\eta_{\infty_-};B)\cdot
\theta(\vec\phi(0,0) +\vec{\Omega}_{0}+\vec K+\vec\eta_{\mathfrak o_+};B)}
{\theta(\sum_{k=1}^2m_k\vec\Omega_{\beta_k}+\vec{\Omega}_{0}+\vec\phi(0,0)
+\vec K+\vec\eta_{\mathfrak o_+};B)\cdot \theta(\vec\phi(0,0)
+\vec{\Omega}_{0}+\vec K+\vec\eta_{\infty_-};B)} \nonumber \\
&\times \exp\sum_{k=1}^2 m_k\int_{\mathfrak o_+}^{\infty_-}\omega[\mathfrak p(\beta_k^2),\infty_+],
\label{eq:4.37}
\\
b(m_{1},m_{2})=&b(0,0)\frac{\theta(\sum_{k=1}^2m_k\vec\Omega_{\beta_k}+\vec\phi(0,0)
+\vec K+\vec{\eta}_{\infty_+};B)\cdot\theta(\vec\phi(0,0)+\vec K+\vec{\eta}_{\mathfrak o_-};B)}
{\theta(\sum_{k=1}^2m_k\vec\Omega_{\beta_k}+\vec\phi(0,0)+\vec K
+\vec{\eta}_{\mathfrak o_-};B)\cdot\theta(\vec\phi(0,0)+\vec K+\vec{\eta}_{\infty_+};B)} \nonumber \\
 &\times (-\beta_{1}^2\gamma_+^{\beta_{1}^2})^{m_{1}} (-\beta_{2}^2
 \gamma_+^{\beta_{2}^2})^{m_{2}}
 \exp \sum _{k=1}^{2}m_{k}\int_{\mathfrak o_-}^{\mathfrak p^{\prime}_0}
 \omega[\mathfrak p(\beta_{k}^2),\infty_+],\label{eq:4.38}
\end{align}
\end{subequations}
where the dipole differential $\omega[\mathfrak p(\beta_k^2),\infty_+]$ is defined as
 \begin{equation}\label{eq:4.40}
\omega[\mathfrak p(\beta_k^2),\infty_+]=
\left(\zeta^{g}+\frac{\xi+\sqrt{R(\beta_{k}^{2})}}{\zeta-\beta_{k}^{2}}\right)
\frac{\mathrm{d}\zeta}{2\sqrt{R(\zeta)}},
\end{equation}
with $\xi$ and $R(\zeta)$ given by (\ref{eq:3.10}). Moreover,
 \begin{equation}\label{eq:4.41}
\gamma_+^{\beta_{k}^2}=\lim_{\mathfrak p\to \infty_+}\frac{1}{\zeta}
\exp\int_{\mathfrak p^{\prime}_0}^{\mathfrak p} \omega[\mathfrak p(\beta_{k}^2),\infty_+].
\end{equation}
\end{theorem}

In what follows, we  illustrate the quasi-periodic behavior of the solutions by the case of $g=1$.
The ingredients  of this case are the following:
\begin{itemize}
\item
{The elliptic curve reads
\begin{equation}\label{eq:4.42}
\xi^2=R(\zeta)=(\zeta-\zeta_{1})(\zeta-\zeta_{2})(\zeta-\zeta_{3})(\zeta-\zeta_{4}),
\end{equation}
where $\zeta_{1}<\zeta_{2}<\zeta_{3}<\zeta_{4}$.}
\item
{The normalized holomorphic differential is
\begin{equation}\label{eq:4.43}
\omega_1=\displaystyle \frac{C_{11}\mathrm{d}\zeta}{2\sqrt{(\zeta-\zeta_{1})(\zeta-\zeta_{2})(\zeta-\zeta_{3})(\zeta-\zeta_{4})}},
\end{equation}
where the normalization constant $C_{11}$ is
\begin{equation}\label{eq:4.44}
C_{11}^{-1}=\int_{a_1}\displaystyle \frac{\mathrm{d}\zeta}{2\sqrt{(\zeta-\zeta_{1})(\zeta-\zeta_{2})(\zeta-\zeta_{3})(\zeta-\zeta_{4})}}.
\end{equation}
}
\item
{Two periods are
\begin{equation}\label{eq:4.45}
1=\int_{a_1}\omega_1, \ \ B_{11}=\int_{b_1}\omega_1.
\end{equation}
}
\item
{The Jacobi theta function $\vartheta_{3}$ is
\begin{equation}\label{eq:C11}
\vartheta_{3}(z\mid B_{11})=\theta(z; B_{11})=\sum_{n=-\infty}^{+\infty}
\exp[\pi\sqrt{-1}(n^{2}B_{11}+2nz)],\ \ z\in \mathbb{C}.
\end{equation}
}
\end{itemize}

Thus, when the genus is $g=1$, by introducing
\begin{subequations}\label{eq:4.48}
\begin{align}
&\Omega_{\beta_k}=\int_{\mathfrak p(\beta_{k}^2)}^{\infty_+}\omega_1,\ \ 
\Omega_{0}=\int_{\mathfrak o_-}^{\infty_+}\omega_1,\ \  
K_1=\displaystyle\frac{1+B_{11}}{2}, \label{eq:4.48a}\\
& \eta_{\infty_\pm}=-\int_{\mathfrak{p}_{0}}^{\infty_\pm}\omega_1,\ \ 
\eta_{\mathfrak o_\pm}=-\int_{\mathfrak{p}_{0}}^{\mathfrak o_\pm}\omega_1,\label{eq:4.48b}\\
&\omega[\mathfrak p(\beta_k^2),\infty_+]
=
\left(\zeta+\frac{\xi+\sqrt{R(\beta_{k}^{2})}}{\zeta-\beta_{k}^{2}}\right)
\frac{\mathrm{d}\zeta}{2\sqrt{R(\zeta)}}, \label{eq:4.48c}\\
&\gamma_+^{\beta_{k}^2}=\lim_{\mathfrak p\to \infty_+}\frac{1}{\zeta}
\exp\int_{\mathfrak p^{\prime}_0}^{\mathfrak p} \omega[\mathfrak p(\beta_{k}^2),\infty_+]\label{eq:4.48d},
\end{align}
\end{subequations}
the algebro-geometric solutions (\ref{eq:4.37}) and (\ref{eq:4.38}) can be written as
\begin{subequations}\label{4.13}
\begin{align}
a(m_1,m_2)=&a(0,0)\mathfrak{a}(m_1,m_2)
\exp\sum_{k=1}^2 m_k\int_{\mathfrak o_+}^{\infty_-}\omega[\mathfrak p(\beta_k^2),\infty_+],\\
b(m_{1},m_{2})=&b(0,0)\mathfrak{b}(m_1,m_2)
\left(-\beta_{1}^2\gamma_+^{\beta_{1}^2}\right)^{m_{1}}
\left(-\beta_{2}^2\gamma_+^{\beta_{2}^2}\right)^{m_{2}}
\exp \sum _{k=1}^{2}m_{k}\int_{\mathfrak o_-}^{\mathfrak p^{\prime}_0}
\omega[\mathfrak p(\beta_{k}^2),\infty_+],
\end{align}
\end{subequations}
where
\begin{subequations}\label{4.14}
\begin{align}
\mathfrak{a}(m_1,m_2)&=\frac{\theta(\sum_{k=1}^2m_k\Omega_{\beta_k}+\Omega_{0}+\phi(0,0)
+ K_{1}+\eta_{\infty_-};B_{11})\cdot
\theta(\phi(0,0) +\Omega_{0}+ K_{1}+\eta_{\mathfrak o_+};B_{11})}
{\theta(\sum_{k=1}^2m_k\Omega_{\beta_k}+\Omega_{0}+\phi(0,0) +K_{1}+
\eta_{\mathfrak o_+};B_{11})
\cdot \theta(\phi(0,0)+\Omega_{0}+ K_{1}+\eta_{\infty_-};B_{11})},\label{eq:CLLa}\\
\mathfrak{b}(m_1,m_2)&=\frac{\theta(\sum_{k=1}^2m_k\Omega_{\beta_k}+\phi(0,0)+K_{1}
+\eta_{\infty_+};B_{11})\cdot\theta(\phi(0,0)+K_{1}+\eta_{\mathfrak o_-};B_{11})}
{\theta(\sum_{k=1}^2m_k\Omega_{\beta_k}+\phi(0,0)+ K_{1}+\eta_{\mathfrak o_-};B_{11})
\cdot\theta(\phi(0,0)+ K_{1}+\eta_{\infty_+};B_{11})}\label{eq:CLLb}.
\end{align}
\end{subequations}
To illustrate the solutions, we set
\begin{equation}\label{zeta-j}
\zeta_1=1,~\zeta_2= 3,~\zeta_3=5,~\zeta_4=8,~ \phi(0,0)+ K_{1}=0,~ \beta_1^2=10,~\beta_2^2=13. 
\end{equation}
$\mathfrak{p}(\beta_k^2)=\left(\beta_k^2, \sqrt{R(\beta_k^2)}\right)$  for $k=1,2$
denote the  points on the elliptic curve defined by \eqref{eq:4.42}.
The quasi-periodic evolution of the solutions are revealed by the components
$\mathfrak{a}(m_1,m_2)$ and $\mathfrak{b}(m_1,m_2)$.
The integrals \eqref{eq:4.44}, \eqref{eq:4.45} and \eqref{eq:4.48} involved in these components are numerically evaluated using Mathematica, yielding the results
\begin{align}\label{parameters}
&C_{11}=1.30467\,\mathrm{i}, ~~   B_{11}=1.21091\,\mathrm{i}, ~~
\Omega_{0}=-0.485654\,\mathrm{i}, ~~ \Omega_{\beta_1}=0.123456\,\mathrm{i}, ~~ \Omega_{\beta_2}=0.0769913\,\mathrm{i}, \nonumber\\
&\eta_{\mathfrak o_+}=-0.084846\,\mathrm{i}, ~~ \eta_{\mathfrak o_-}=-0.391547\,\mathrm{i}, ~~ \eta_{\infty_+}=0.0940414\,\mathrm{i},
~~ \eta_{\infty_-}=-0.57037\,\mathrm{i},
\end{align}
where we take $\mathfrak{p}_0=(-3, \sqrt{R(-3}))$.
We depicted such $\mathfrak{a}(m_1,m_2)$ and $\mathfrak{b}(m_1,m_2)$ in Figure  \ref{Fig-1} .
\begin{figure}[!h]
\centering
\subfigure[]{
\begin{minipage}{5.5cm}
\includegraphics[width=\textwidth]{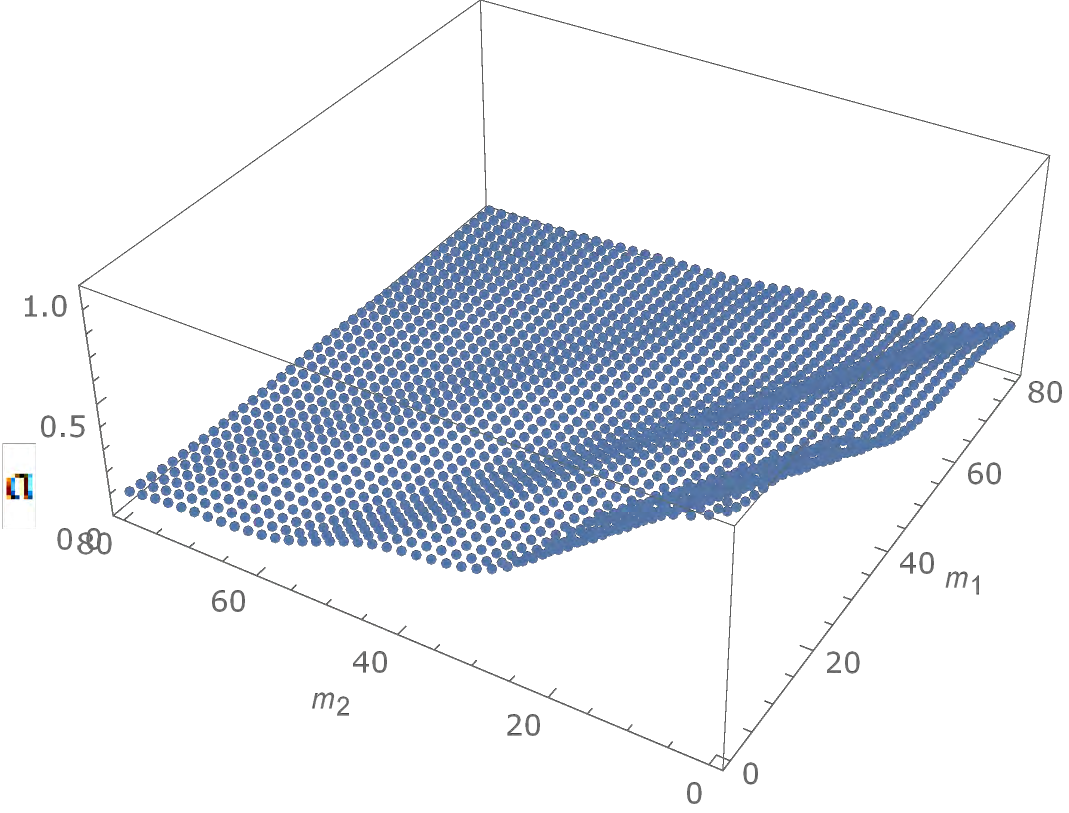}
%\centering
%\footnotesize{(a)}
\end{minipage}
}
\hspace{0.2in}
\subfigure[]{
\begin{minipage}{5.5cm}
\includegraphics[width=\textwidth]{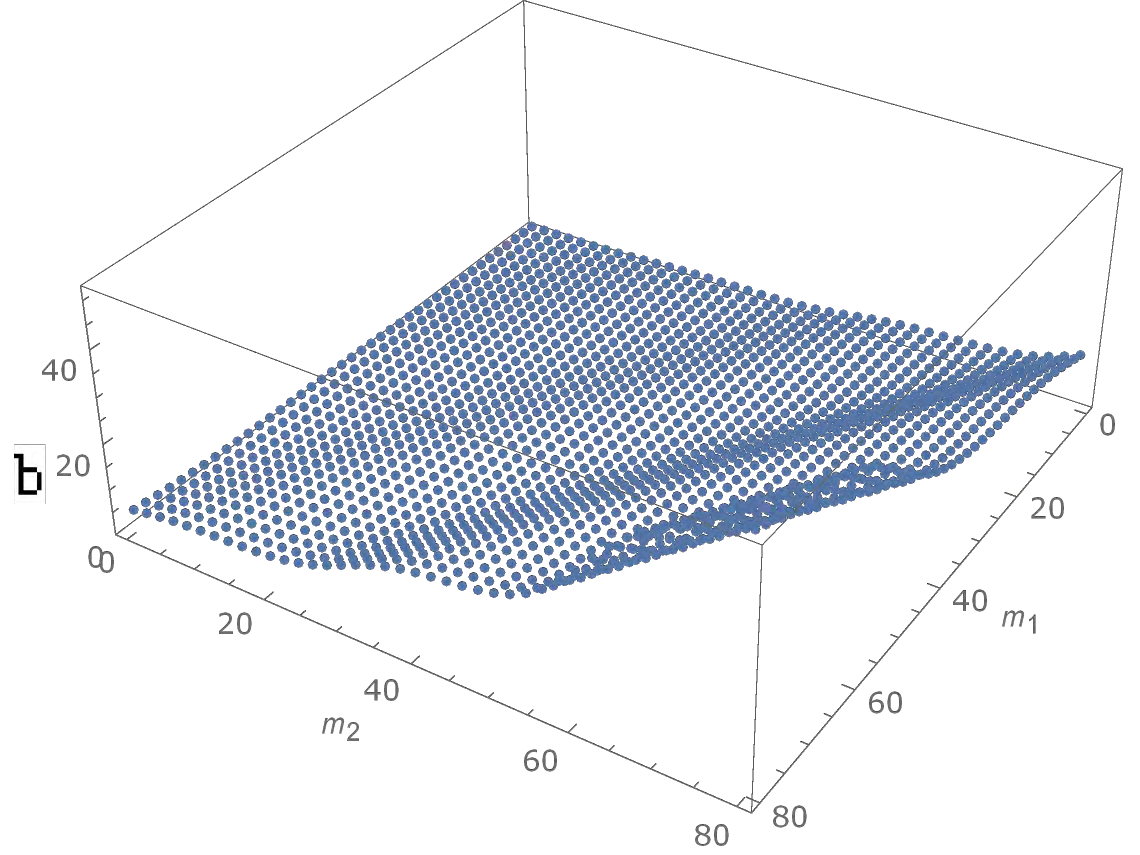}
%\centering
%\footnotesize{(b)}
\end{minipage}
}
\caption{(a) Shape and motion of $\mathfrak{a}(m_1,m_2)$ given in \eqref{eq:CLLa}. 
(b) Shape and motion of $\mathfrak{b}(m_1,m_2)$ given in \eqref{eq:CLLb}.  }
\label{Fig-1}
\end{figure}

\section{Conclusions}\label{sec-5}

In this paper we have constructed algebro-geometric solutions for the discrete CLL system \eqref{eq:1.1}.
These solutions are presented as in \eqref{eq:4.37} together with \eqref{eq:4.38} in terms of the  Riemann theta functions.
Recalling the whole procedure,
the first step is to find an integrable symplectic map,
which is usually associated to the spectral problem of the investigated equation.
In this paper, the spectral problem \eqref{Lax-1} is not suitable for constructing such a map.
In stead, what we employed is \eqref{eq:3.23b}, which is associated with \eqref{Lax-1}
via certain gauge transformations (see \cite{LZZ}) and deformations.
To prove the map \eqref{eq:3.23} is integrable and symplectic,
we introduced the Lax matrix $L(\lambda)$ in \eqref{eq:3.2}
and made use of $r$-matrix and the constraints \eqref{eq:3.25}.
Note that these constraints make the map no longer linear
but they do bridge the gap between the eigenfunctions and potentials.
In the second step, by means of the map we introduced evolutions of the  eigenfunctions (the discrete phase flows).
Then we introduced Baker-Akhiezer functions and figured out their devisors.
After introducing Abel-Jacobi variables, we were able to express the Baker-Akhiezer functions
using the  Riemann theta functions,
which finally lead to a reconstruction of the discrete potentials $(a(m), b(m))$  in terms of the  Riemann theta functions,
which are presented in \eqref{eq:4.34}.
In the third step, we introduced evolution of the discrete phase flow in another direction
(associated with another parameter).
The compatibility of the maps $S^{m_{1}}_{\beta_{1}}$ and $S_{\beta_{2}}^{m_{2}}$
gives rise to the discrete CLL system \eqref{eq:1.1} as a necessary condition.
As a result, we could present $(a(m_1,m_2), b(m_1.m_2))$ in terms of the  Riemann theta functions,
as given in \eqref{eq:4.37} together with \eqref{eq:4.38},
which are the algebro-geometric solutions for the discrete CLL system \eqref{eq:1.1}.

In the final part of the paper, we raise several related problems for further consideration.
One is about high genus solutions.
In this paper we only illustrated the $g=1$ case. The high genus solutions are still waiting for understanding.
In addition, in this paper the evolution of the phase flow is introduced by iterating the
integrable symplectic map $S_{\beta}$, which means in the solutions the independent variables
$m_1$ and $m_2$ are nonnegative integers.
Whether  $(m_1,m_2)$ can be extended to $\mathbb{Z}\times \mathbb{Z}$,
one needs to check whether the inverse of the  map $S_{\beta}$ is still integrable and symplectic.
This is actually a common problem appearing in solving discrete equations through the
nonlinearizing Lax approach.
The third problem is about reduction of solutions.
Note that the discrete CLL system \eqref{eq:1.1} fails
to get a discrete analogue of the CLL equation \eqref{CLL} by the reduction of the type $b=a^*$.
However, it can be reduced to the discrete Burgers equation \cite{Z-Burgers-2022}
by taking $a=1$ (or $b=1$) \cite{CZZ-2021}.
Whether one can reduce the algebro-geometric solutions to those of the discrete Burgers equation
will be an interesting problem.
The last one is about the finite-gap integration. So far the procedure developed
by Novikov and Matveev and their collaborators
\cite{Dubrovin-1975,DN-1975,IM-1975a,IM-1975b,Krichever-1977, Novikov-1974}
is not yet discretized.

\vskip 20pt

\subsection*{Acknowledgments}
This project is supported by the NSF of China (Nos. 12271334, 12326428 and 12301309), 
the Key Scientific Research Project of Colleges and Universities in Henan Province (No. 23A110015),
and the Basic Research and Cultivation Foundation of Zhengzhou University (No. JC23153002).

\appendix

\section{Discrete mKP equation and its solutions}\label{A-1}

Here, we investigate the discrete mKP equation \eqref{eq:A-1-1} by using the discrete CLL spectral problem.
In addition to \eqref{eq:2.6}, let us consider
\begin{equation}\label{eq:A-0}
\widehat{\chi}=D^{(\beta_3)}\chi=\frac{1}{\sqrt{\Upsilon^{(3)}}}\left(\begin{array}{cc}
        -\lambda^2+\beta_{3}^2+a\widehat{b}&\lambda a \\
        \lambda \widehat{b}&\beta^2_3
        \end{array}\right)\chi,\ \ \Upsilon^{(3)}=\beta_{3}^2+a\widehat{b}
\end{equation}
where $\beta_3$ is the lattice parameters associated with $m_{3}$-direction, and
the shift in this direction is denoted by a hat, i.e. $\widehat{f}=f(m_{1},m_{2},m_{3}+1)$.
From the compatibility between \eqref{eq:2.6a} and \eqref{eq:2.6b}, we have
\begin{equation}\label{eq:A-0-1}
Z_{12}\equiv\overline{\Upsilon}^{(1)}+{\Upsilon}^{(2)}-\t{{\Upsilon}}^{(2)}
-{\Upsilon}^{(1)}+\t{a}\t{b}-\overline{a}\overline{b}=0.
\end{equation}
Similarly,  using \eqref{eq:2.6a}, \eqref{eq:2.6b} and \eqref{eq:A-0}  we have
\begin{align}
Z_{23} & \equiv \widehat{\Upsilon}^{(2)}+{\Upsilon}^{(3)}-\overline{{\Upsilon}}^{(3)}
-{\Upsilon}^{(2)}+\overline{a}\overline{b}-\widehat{a}\widehat{b}=0, \label{eq:A-0-2} \\
Z_{31} & \equiv  \t{\Upsilon}^{(3)}+{\Upsilon}^{(1)}-\widehat{{\Upsilon}}^{(1)}-{\Upsilon}^{(3)}
+\widehat{a}\widehat{b}-\t{a}\t{b}=0. \label{eq:A-0-3}
\end{align}
We introduce function $W=W(m_1,m_2,m_3)$ by
\begin{equation}\label{eq:A-0-0}
W-\widetilde{W}=\ln \frac{\Upsilon^{(1)}}{\beta^2_1},\ \
W-\overline{W}=\ln \frac{\Upsilon^{(2)}}{\beta^2_2}, \ \
W-\widehat{W}=\ln \frac{\Upsilon^{(3)}}{\beta^2_3}
\end{equation}
and from the identity
\begin{equation*}
Z_{12}+Z_{23}+Z_{31}=0,
\end{equation*}
we can obtain the discrete mKP equation \eqref{eq:A-1-1}.
In other words, $W$ defined by \eqref{eq:A-0-0} satisfies the discrete mKP equation \eqref{eq:A-1-1}.

In addition, from equation \eqref{eq:4.35}, the function $W_m$ that satisfies $W_{m+1}-W_{m}=-\ln{\frac{\Upsilon_m}{\beta^2}} $ can be worked out:
\begin{align}
    W_m=&W_0-\ln{[(-1)^m\beta^{-2m}\frac{\theta(\vec\psi(0)+\vec K+\vec{\eta}_{\infty_+};B)\theta(\vec\psi(m)+\vec K+\vec{\eta}_{\mathfrak{o_+}};B)}
{\theta(\vec\psi(m)+\vec K+\vec{\eta}_{\infty_+};B)\theta(\vec\psi(0)+\vec K+\vec{\eta}_{\mathfrak{o}_+};B)}(\gamma_+^{\beta^2})^{-m}]} \nonumber \\
&+m\int_{\mathfrak o_+}^{\mathfrak{p^{\prime}_0}}\omega[\mathfrak p(\beta^2),\infty_+].
\label{eq:A-1}
\end{align}
As a result, we have the explicit solution for discrete mKP equation: 
\begin{subequations}
\begin{equation}\label{eq:A-2}
W(m_1,m_2,m_3)={W}_{2}(m_1,m_2,m_3)+{W}_{1}(m_1,m_2,m_3),
\end{equation}
where %${W}_i={W}_i(m_1,m_2,m_3)$,
\begin{align}
{W}_{2}&=\ln\Big[\frac{\theta(\sum_{k=1}^3m_k\vec\Omega_{\beta_k}
+\vec{\Omega}_{0}+\vec\phi(0,0,0)+\vec K+\vec{\eta}_{\infty_+};B)
\theta(\vec\phi(0,0,0)+\vec{\Omega}_{0}+\vec K+\vec{\eta}_{\mathfrak o_+};B)}
{\theta(\sum_{k=1}^3m_k\vec\Omega_{\beta_k}+\vec{\Omega}_{0}
+\vec\phi(0,0,0)+\vec K+\vec{\eta}_{\mathfrak o_+};B)
\theta(\vec\phi(0,0,0)+\vec{\Omega}_{0}+\vec K+\vec{\eta}_{\infty_+};B)}\Big],\label{eq:A-3}\\
{W}_{1}&=\sum_{k=1}^3 m_k\left(\int_{\mathfrak o_+}^{\mathfrak{p^{\prime}_0}}
\omega[\mathfrak p(\beta_k^2),\infty_+]
+\ln(-\beta_k^2\gamma_+^{\beta_k^2})\right)+W(0,0,0).\label{eq:A-4}
\end{align}
\end{subequations}
This solution has a same structure as the one derived in \cite{XCZ-2022}. 
In the case of genus $g=1$, the quasi-periodic evolution of $\mathfrak{W}_2(m_1,m_2,m_3)$ 
for the parameters in \eqref{parameters} and $\Omega_{\beta_3}=0.068633\,\mathrm{i}$ with $\beta_3^2=14$
is  illustrated in Figure \ref{Fig-2}.
\begin{figure}[!h]
\centering
\begin{minipage}{6cm}
\includegraphics[width=\textwidth]{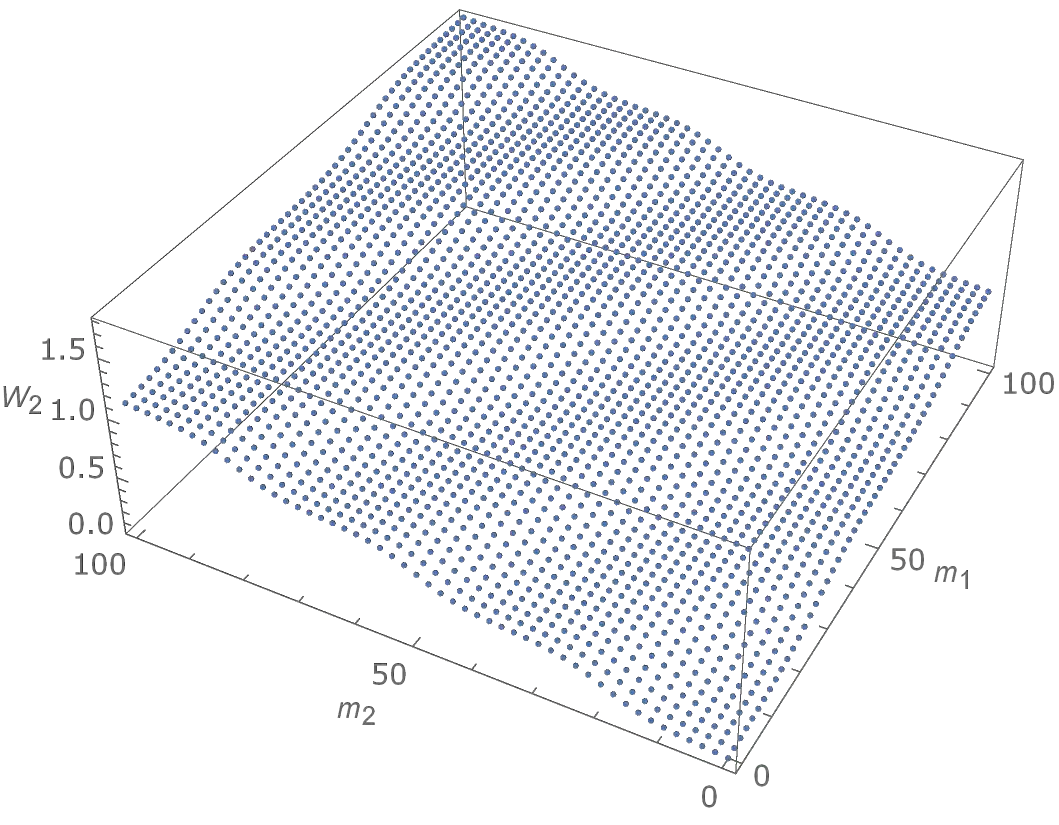}
\centering
\end{minipage}
\caption{Shape and motion of $W_2(m_1,m_2,0)$.   }
\label{Fig-2}
\end{figure}

%\small


\begin{thebibliography}{90}


\bibitem{Arnold} V.I. Arnold,
            Mathematical Methods in Classical Mechanics,
            Springer,  Berlin,  1978.


\bibitem{BBEIM-1994}  E.D. Belokolos, A.I. Bobenko, V.Z. Enolskij, A.R. Its, V.B. Matveev,
         Algebro-geometric Approach to Nonlinear Integrable Equations,
         Spinger, Berlin, 1994.

\bibitem{Bruschi} M. Bruschi, O. Ragnisco, P.M. Santini, G.Z. Tu,
         Integrable symplectic maps,
         Physica D, 49 (1991) 273-294.

\bibitem{Cao1990} C.W. Cao,
        Nonlinearization of the Lax system for AKNS hierarchy,
        Sci. China Ser. A, 33 (1990) 528-536.

\bibitem{Cao-G-1990} C.W. Cao, X.G. Geng,
       Classical integrable systems generated through nonlinearization of eigenvalue problems,
       in Nonlinear Physics, Research Reports in Physics,
       C.H. Gu, Y.S. Li, G.Z. Tu (eds.),
       Springer-Verlag, Berlin, 1990, pp66-78.

\bibitem{Cao-WG-1999} C.W. Cao, Y.T. Wu, X.G. Geng,
       Relation between the Kadometsev-Petviashvili equation and the confocal involutive system,
       J. Math. Phys., 40 (1999) 3948-3970.

\bibitem{Cao-WG-1999b} C.W. Cao, Y.T. Wu, X.G. Geng,
      From the special 2+1 Toda lattice to the Kadomtsev-Petviashvili equation,
      J. Phys. A: Math. Gen., 32 (1999) 8059-8078.



\bibitem{CaoX-JPA-2012} Cao, C.W.,  Xu, X.X.:
          A finite genus solution of the H1 model.
          J. Phys. A: Math. Theor., 45 (2012)  055213 (13pp).

\bibitem{CXZ-2020} C.W. Cao, X.X. Xu, D.J. Zhang,
           On the lattice potential KP equation,
           in            Asymptotic, Algebraic and Geometric Aspects of Integrable Systems,
           Springer Proceedings in Mathematics $\&$ Statistics, vol. 338,
           F.W. Nijhoff, Y. Shi, D.J. Zhang (eds.),
           Springer, Cham, 2020, pp.213-237.

\bibitem{CZ2012a} C.W. Cao, G.Y. Zhang,
             A finite genus solution of the Hirota equation via integrable symplectic maps,
             J. Phys. A: Math. Theor.,   45 (2012) 095203 (25pp).

\bibitem{CZ2012b} C.W. Cao, G.Y. Zhang,
           Integrable symplectic maps associated with the ZS-AKNS spectral problem,
           J. Phys. A: Math. Theor., 45 (2012)  265201 (15pp).

\bibitem{Chen-2002} D.Y. Chen,
            $k$-constraint for the modified Kadomtsev--Petviashvili system.
            J. Math. Phys.,  43 (2002) 1956-1965.

\bibitem{CLL} H.H. Chen, Y.C. Lee, C.S. Liu,
           Integrability of non-linear Hamiltonian systems by inverse scattering method,
           Phys. Scr., 20 (1979) 490-492.

\bibitem{CZZ-2021} K. Chen, C. Zhang, D.J. Zhang,
         Squared eigenfunction symmetry of the D$\Delta$mKP hierarchy and its constraint,
         Stud. Appl. Math., 147 (2021) 752-791.

\bibitem{DJM-1983} E. Date, M. Jimbo, T. Miwa,
           Method for generating discrete soliton equations: IV,
           J. Phys. Soc. Japan, 52 (1983) 761-765.


\bibitem{Dubrovin-1975} B.A. Dubrovin,
      Periodic problems for the Korteweg--de Vries equation in the class of finite band potentials,
      Funct. Anal. Appl., 9 (1975) 215-223.

\bibitem{DN-1975} B.A. Dubrovin, S.P. Novikov,
      Periodic and conditionally periodic analogs of the many soliton solutions of the Korteweg-de Vries equation,
      Sov. Phys. JETP, 40 (1975) 1058-1063.


\bibitem{Farkas} H.M. Farkas, I. Kra,
            Riemann Surfaces,
            Springer, New York, 1992.


\bibitem{GDZ-2007} X.G. Geng, H.H. Dai, J.Y. Zhu,
      Decomposition of the discrete Ablowitz-Ladik hierarchy,
      Stud. Appl. Math., 118 (2007) 281-312.

\bibitem{GWH-2013}  X.G. Geng, L.H. Wu, G.L. He,
           Quasi-periodic solutions of the Kaup--Kupershmidt hierarchy,
           J. Nonlinear Sci., 23 (2013) 527-555.

\bibitem{GZD-2014} X.G.  Geng, Y.Y. Zhai, H.H. Dai,
      Algebro-geometric solutions of the coupled modified Korteweg-de Vries hierarchy,
      Adv. Math., 263 (2014)  123-153.


\bibitem{Griffiths} J.P. Griffiths,  J. Harris,
         Principles of Algebraic Geometry,
         Wiley, New York, 1978.

\bibitem{Goldstein} H. Goldstein, %C. Poole, J. Safko,
         Classical Mechanics, 2nd edition,
         Addison-Wesley Publ., San Francisco, 1980.

\bibitem{Jordan} T.F. Jordan,
         Steppingstones in Hamiltonian dynamics,
         Amer. J. Phys.,  72 (2004) 1095-1099.


\bibitem{IM-1975a} A.R. Its, V.B. Matveev,
         Hill operators with a finite number of lacunae,
         Funct. Anal. Appl.,  9 (1975)  65-66.

\bibitem{IM-1975b} A.R. Its, V.B. Matveev,
         Schr\"{o}dinger operators with the finite-band spectrum and
         the N-soliton solutions of the Korteweg--de Vries equation,
         Theor. Math. Phys.,  23 (1975) 343-355.

\bibitem{Krichever-1977} I.M. Krichever,
         Methods of algebraic geometry in the theory of non-linear equations,
         Russ. Math. Surv.,  32 (1977) 185-213.

\bibitem{LZZ} J. Liu, D.J. Zhang, X.H. Zhao,
        Symmetries of the D$\Delta$mKP hierarchy and their continuum limits,
        Stud. Appl. Math., 152 (2024) 404-430.

\bibitem{Matveev} V.B. Matveev,
          30 years of finite-gap integration theory,
          Phil. Trans. R. Soc. A,   366 (2008)  837-875.

\bibitem{Mumford} D. Mumford,
          Tata Lectures on Theta. I,
          Birkh\"auser, Boston, 1983.

\bibitem{NCW-LNP-1985}  F.W. Nijhoff, H. Capel, G. Wiersma,
         Integrable lattice systems in two and three dimensions,
         in: R. Martini (ed.), Geometric Aspects of the Einstein Equations and Integrable Systems,
         Lecture Notes in Physics, Vol.239,
         Springer, Berlin, Heidelberg,  1985, pp263-302.
         %https://doi.org/10.1007/3-540-16039-6_8

\bibitem{Nijhoff} F.W. Nijhoff, H. Capel, G. Wiersma, G.R.W. Quispel,
           B\"{a}cklund transformations and three-dimensional lattice equations,
           Phys. Lett.  A, 105 (1984) 267-272.

\bibitem{Novikov-1974} S.P. Novikov,
         The periodic problem for the Korteweg--de vries equation,
         Funct. Anal. Appl.,  8 (1974) 236-246.

\bibitem{RCW-1995} O. Ragnisco, C.W. Cao, Y.T. Wu,
       On the relation of the stationary Toda equation and the symplectic maps,
       J. Phys. A: Math. Gen., 28 (1995) 573-588.

\bibitem{Toda} M. Toda,
         Theory of Nonlinear Lattices,
         Springer, Berlin, 1981.

 \bibitem{Veselov} A.P. Veselov,
         Integrable maps,
         Russ. Math. Surv.,  46 (1991) 3-45.

\bibitem{Veselov1} A.P. Veselov,
         What is an integrable mapping?
         In: V.E. Zakharov (ed.), What Is Integrability?
         Springer-Verleg, Berlin, Heidelberg, 1991, pp251-272.

\bibitem{Suris} Yu. B. Suris,
         The Problem of Integrable Discretization: Hamilltonian Approach,
         Birkh\"auser, Basel, 2003.

\bibitem{XCN-2021} X.X. Xu, C.W. Cao, F.W. Nijhoff,
         Algebro-geometric integration of the Q1 lattice equation via nonlinear integrable symplectic maps,
         Nonlinearity,  34 (2021) 2897-2918.

\bibitem{XCZ-2022} X.X. Xu, C.W. Cao, D.J. Zhang,
        Algebro-geometric solutions to the lattice potential modified Kadomtsev-Petviashvili equation,
        J. Phys. A: Math. Theor., 55  (2022) 375201 (28pp).

\bibitem{XCZ-JNMP-2020} X.X. Xu, C.W. Cao, G.Y. Zhang,
         Finite genus solutions to the lattice Schwarzian Korteweg-de Vries equation,
         J. Nonl. Math. Phys.,    27 (2020) 633-646.

\bibitem{XJN-2021} X.X. Xu, M.M. Jiang, F.W. Nijhoff,
         Integrabe symplectic maps associated with discrete Korteweg--de Vries-type equations,
         Stud. Appl. Math.,  146 (2021)  233-278.

\bibitem{Z-Burgers-2022} D.J. Zhang,
         The discrete Burgers equation,
         Part. Diff. Equa. Appl. Math., 5 (2022) 100362 (5pp).

\bibitem{Zhou-1997} R.G. Zhou,
       The finite-band solution of Jaulent-Miodek equation,
       J. Math. Phys., 38 (1997) 2335-2546.



\end{thebibliography}
\end{document}